\documentclass{amsart}

\usepackage[utf8]{inputenc}
\usepackage{amsmath}
\usepackage{amsfonts}
\usepackage{mathrsfs}
\usepackage{graphicx}
\usepackage{caption}
\usepackage{subcaption}
\usepackage{lineno}

\modulolinenumbers[5]

\newtheorem{theorem}{Theorem}
\newtheorem{lemma}{Lemma}

\newtheorem{definition}{Definition}
\newtheorem{corollary}{Corollary}

\newcommand{\R}{\mathbb{R}}

\def \R    {\mathbb{R}}    
\def \P    {\mathbb{P}}    
\def \SO   {\mbox{SO}}     
\def \Rot  {\mathscr{R}}   

\def \Cl   {\mathrm{C}\ell}
\def \e    {\mathbf{e}}

\DeclareMathOperator{\Pin}{Pin}
\DeclareMathOperator{\Spin}{Spin}

\def \S    {\mathbf{S}}    

\def \GG   {\mathcal{G}}   

\allowdisplaybreaks

\begin{document}


\title[Parameterization of configuration space obstacles]{Parameterization of configuration space obstacles in three-dimensional rotational motion planning}

\author{Przemysław Dobrowolski}
\address{Faculty of Mathematics and Information Science \\ Warsaw University of Technology, Poland}

\begin{abstract}
This study investigates the exact geometry of the configuration space in three-dimensional rotational motion planning. A parameterization of configuration space obstacles is derived for a given triangulated or ball-approximated scene with an object rotating around a given point. The results obtained are an important step towards a complete description of the rich structure of a configuration space in three-dimensional rotational motion planning.
\end{abstract}



\maketitle


\section{Introduction}

There are two major approaches for motion planning: heuristic methods and combinatorial methods, which are also called exact methods. The first are considered practical and fast, whereas the second are usually slower and much more difficult to design and implement. Due to ongoing performance improvements, especially in the area of parallel computation on graphics cards (CUDA and OpenCL), it is becoming more feasible to implement exact motion planners. These motion planners have clear advantages because they are complete, they can handle narrow passages correctly in a configuration space, and they can handle queries about the non-existence of a motion path. Heuristic methods have none of these advantages. Exact motion planners are known for translation-only movements and they have been implemented successfully. In two dimensions, examples include algorithms for convex polygons (\cite{agarwal1999motion} and \cite{journals/dcg/LevenS87a}), simple segments (\cite{LevenSharirEfficientSimpleMotion}), L-shaped objects (\cite{journals/siamcomp/HalperinOS92}), and general polygons (\cite{Guibas1989} and \cite{journals/dcg/HalperinS96}). For three-dimensional motion planning, the Minkowski sum can be used to compute the configuration space, as described by Varadhan et al. \cite{Varadhan3D}. In a more recent study by Asano et al. \cite{AsanoHN02} the authors consider the problem of calculating effecitely the Minkonwski sum for convex polyhedra. A complete survey on the Minkowski sum computation for polyhedra was presented by Weibel \cite{WeibelPhD} in 2007 in his PhD thesis. Exact motion planners with a rotational component are rare. The most important is the algorithm presented by Avnaim et al. \cite{avnaim1988practical} for planning the movement of an arbitrary polygon amidst polygonal obstacles, including both translations and rotations.  Avnaim et al. \cite{avnaim1988practical} also described the parameterization of configuration space obstacles in the case of two-dimensional movement with translations and rotations. Few algorithms have been presented for three dimensions. A translational and rotational motion planning method in three dimensions was presented by Koltun \cite{koltun2005pianos} for a simple segment or a cigar-like object, and it was suggested that this algorithm may be generalized to circles, cylinders, and other general bodies of revolution. Unfortunately, no studies have continued these ideas. General motion planning algorithms also exist such as the Canny algorithm \cite{Can93} and Collins' algorithm \cite{Col75}, but these are only theoretical considerations because they are impractical and they have never been implemented.

The present study is based on the initial theory for three-dimensional rotational motion planners introduced by \cite{phd_dobrowolski}. Assumptions regarding the possible scenes are taken from this previous study. The possible scenes include a triangulated rotating mesh among triangulated obstacle meshes and a ball approximated object in a scene comprising ball-approximated obstacles. It is assumed that an object is always rotating around zero.

The parameterization derived in this study is exact in the selected arithmetic. It does not include an iterative method or approximations. The results are approximate if a floating-point based arithmetic is used. Alternatively, we may use exact numbers with arbitrary precision to obtain mathematically correct results. In this case, it is assumed that $R$ is a unital ring used as a base number format to represent both the scene and configuration space. In practical implementations, $R$ is a ring of integer numbers $Z$, or a field of rational numbers $Q$, depending on which of the two is more effective for describing a given scene. Note that in rotational motion planning, every scene described in $Q$ can be scaled to a scene described in $Z$ but the magnitude of the co-ordinates may grow significantly. When an exact arithmetic is considered, $Z$ and $Q$ are usually implemented with the MPZ and MPQ types from the GMP library. All of the input scalars, vectors, and spinors used in this study are defined in the ring $R$; in particular, all the coordinates of the base vectors $K$, $L$, $A$, $B$ and scalar $c$ are defined in $R$.

All of the results presented in this study have been implemented and they are publicly available from the library \emph{libcs2} \cite{libcs2}. Using the parameterization, it is possible to create an accurate visualization of the configuration space contents and to analyse topological aspects of the configuration space in rotational motion planning. Other possible uses are described in the conclusions of this study.

In this section, we provide a number of definitions and we describe existing results in rotational motion planning. In \cite{phd_dobrowolski}, the author introduced a theory that provides a complete description of the configuration space of a rotating body, where it was shown that both triangulated and ball-only scenes with rotating bodies can be described by one type of predicate: a general (rotational) predicate. The main idea presented in \cite{phd_dobrowolski} is as follows. We describe a triangulated or a ball-approximated scene with a set of collision predicates together with a boolean formula which evaluates true or false if the rotating object collides with the scene or not. The formula consists of all pairs of geometric objects allowed to collide: one from the rotating object and one from the stationary scene.
In the next step, the geometric predicates are identified with surfaces representing them in a configuration space. A single ball-ball collision intorduces 1 surface while a triangle-triangle collision introduces 9 surfaces in the configuration spasce. The set of the surfaces in the configuration space defines an arrangement which is usually quite complicated. The analysis of such arrangements has began with the study presented in \cite{phd_dobrowolski}.

It is assumed that $\Rot_s(v) = s v s^{-1}$ for $s \in \Spin(3)$ is a rotation formula in Clifford algebra, which transforms an input vector $v$ (written in terms of Clifford algebra) into a rotated vector $s v s^{-1}$, according to a rotation represented by a spinor $s$. This is a close equivalent to the quaternion rotation formula $q v q^*$ in quaternion algebra. An isomorphism exists between the spinors and quaternions, which allows all of the results to be transformed into quaternion algebra if necessary.

In \cite{phd_dobrowolski}, it was shown that a triangulated or ball-only scene can be described by a set of general predicates, as follows.
\begin{definition}[A general predicate $\GG$]
Assume that $K$ and $L$ are the ends of a stationary segment, $A$ and $B$ are the ends of a rotating segment, and $c$ is a scalar. The formula:
\begin{equation*}
\GG_s(K, L, A, B, c) = (K \times L) \cdot \Rot_s(A - B) + (K - L) \cdot \Rot_s(A \times B) + c
\end{equation*}
is called \textbf{a general predicate}.
\end{definition}
A major result given in \cite{phd_dobrowolski} was a proof that the problem of computing a configuration space in the case of three-dimensional rotational motions is equivalent to computing an exact arrangement of quadratic surfaces in projective space $\P^3$. Although this provides a tool for computing the configuration space, the only existing algorithm for computing an arrangement of quadrics, which was proposed by Hemmer et al. \cite{journals/jsc/HemmerDPS11}, is computationally expensive and difficult to implement (although it is possible, unlike the Canny \cite{Can93} and Collins' \cite{Col75} algorithms). The arrangement derived in \cite{journals/jsc/HemmerDPS11} includes the parameterization of quadrics intersections and the exact coordinates of intersection graph vertices, which are sufficient for planning a three-dimensional rotational motion. A configuration space is known to be complicated in three-dimensional rotational motion planning in both topological and computational terms. Figure \ref{fig_typical_tt} illustrates an example of a configuration space for a simple scene. The scene comprises a stationary triangle and a triangle allowed to rotate freely in three-dimensional space around the origin. The surfaces represent constraints induced by the collision of the triangles.
\begin{figure}[ht]
\centering
\begin{subfigure}{0.5\textwidth}
  \centering
  \includegraphics[width=0.95\textwidth]{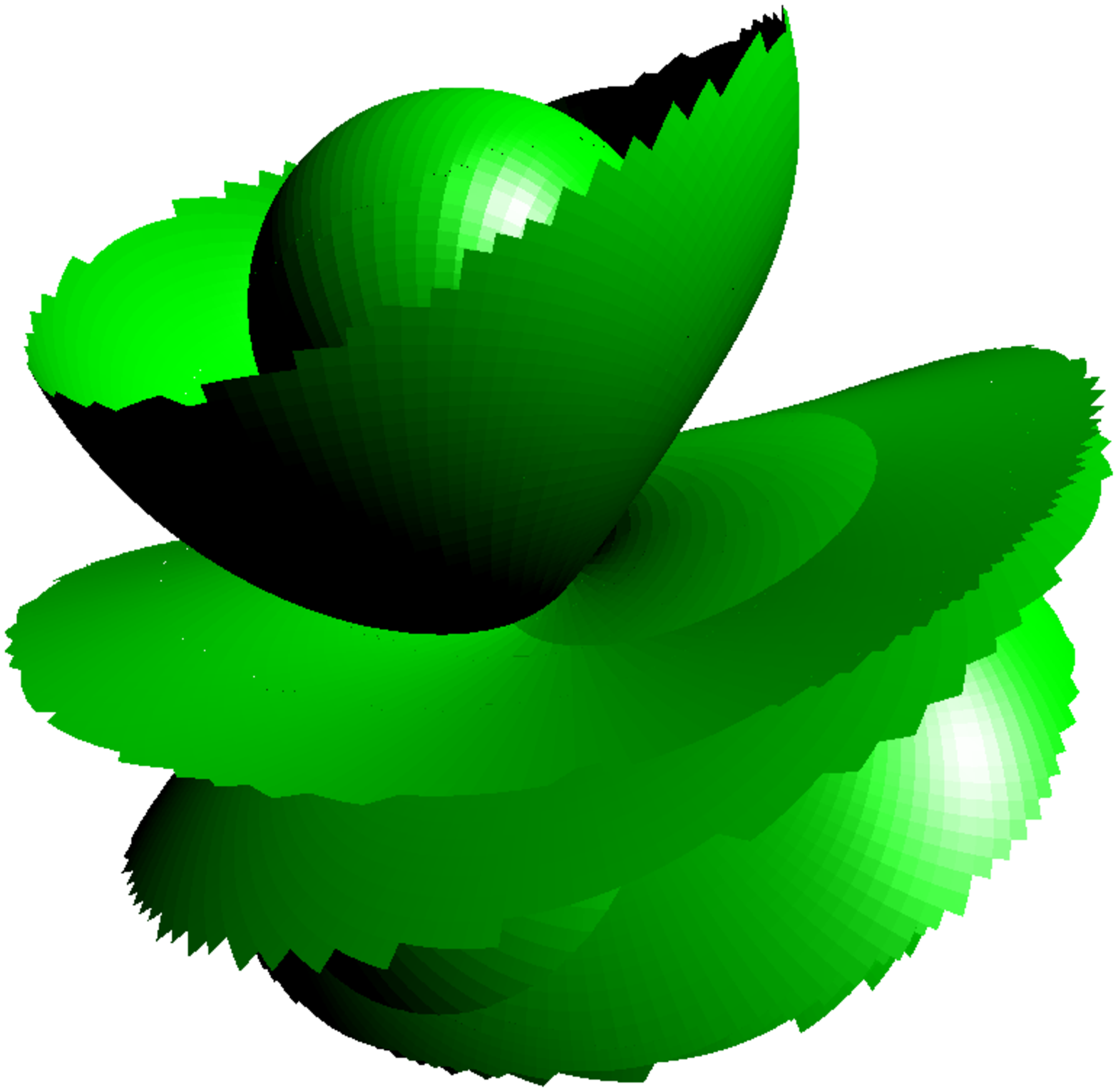}
\end{subfigure}%
\begin{subfigure}{0.5\textwidth}
  \centering
  \includegraphics[width=0.95\textwidth]{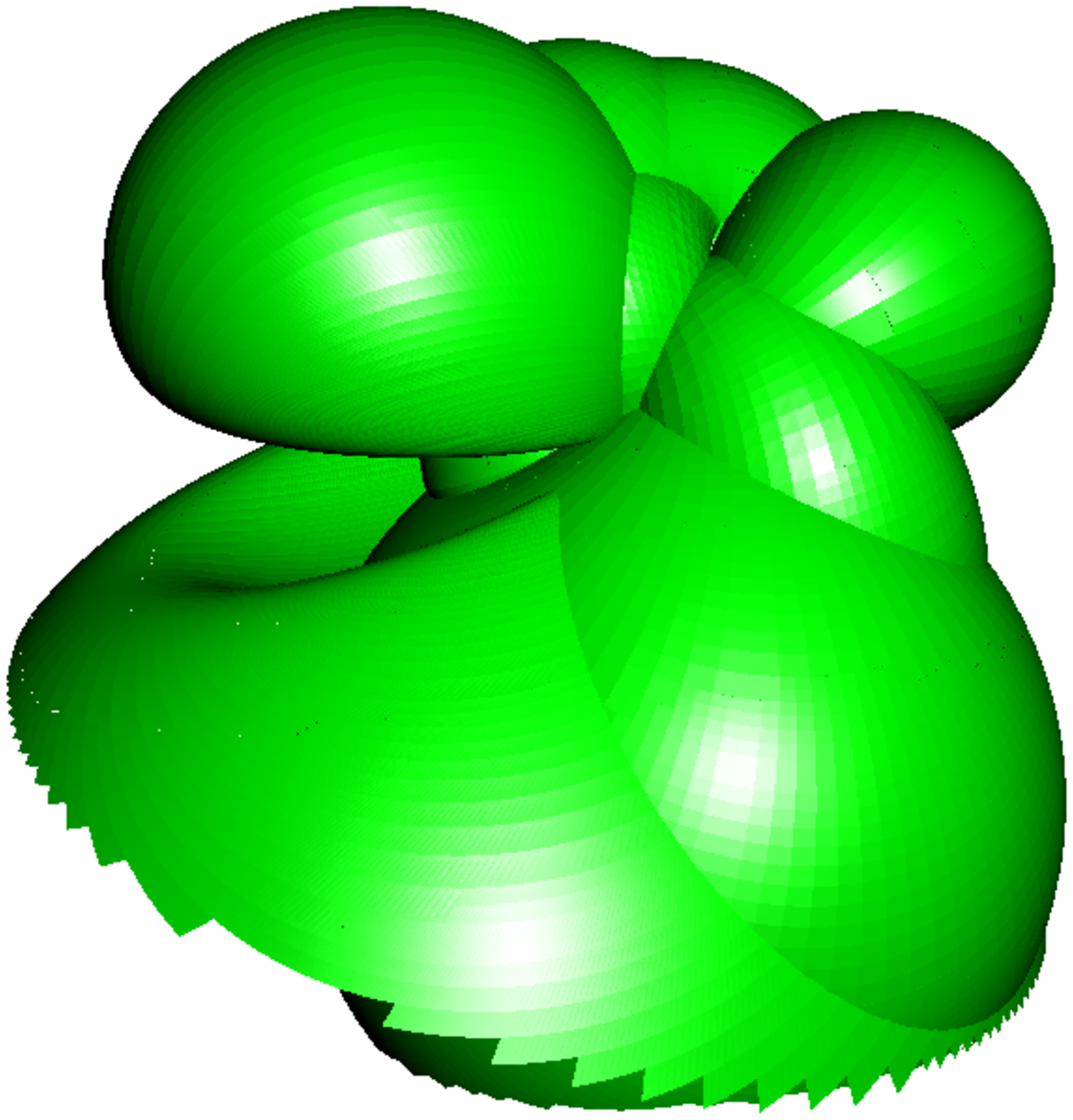}
\end{subfigure}%
\caption{Two examples of the configuration spaces for a simple scene}
\label{fig_typical_tt}
\end{figure}
In general, solutions such as \cite{journals/jsc/HemmerDPS11} are usually slower than the solutions for a specific sub-case of a given problem. Thus, in this study, an alternative method is proposed for deriving an exact parameterization of the configuration space for three-dimensional rotational motion. The idea is to derive a parameterization directly for the surface represented by a general predicate, unlike the algorithm \cite{journals/jsc/HemmerDPS11}, which gives a parameterization for any surface.

In this study, it is assumed that $\GG_s(K, L, A, B, c)$ is a general predicate over a given ring $R$. A general predicate $\GG_s$ in $\Spin(3)$ configuration space induces a quadratic surface $S_\GG(s)$, called \emph{a spin-quadric}. It was proved that a spin-quadric $S_\GG(s)$ is a quadratic form in $\Spin(3)$. Assuming that $s = s_{12} \e_{12} + s_{23} \e_{23} + s_{31} \e_{31} + s_0$ is a spinor with coefficients $s_{12}$, $s_{23}$, $s_{31}$, $s_0$, and $\mathbf{s} = [s_{12}, s_{23}, s_{31}, s_0]^T$ is a related vector-like form of spinor $s$, then we have the following.
\begin{theorem}[Dobrowolski, \cite{phd_dobrowolski}] \label{thm_spin_matrix}
A general predicate $\GG(K, L, A, B, c)$ can be reduced to a quadratic form in $\Spin(3)$, which represents a spin-quadric $S_\GG(s)$ that can be expressed by:
\begin{equation*}
S_\GG(s) = \mathbf{s}^T M_\GG \mathbf{s} = 0
\end{equation*}
with the matrix
\begin{equation*}
M_\GG =
\left[ \begin{smallmatrix}
a_{11} & a_{12} & a_{13} & a_{14} \\
a_{12} & a_{22} & a_{23} & a_{24} \\
a_{13} & a_{23} & a_{33} & a_{34} \\
a_{14} & a_{24} & a_{34} & a_{44}
\end{smallmatrix} \right]
\end{equation*}
or equivalently:
\begin{align*}
S_\GG(s) & = a_{11} s_{12}^2 + a_{22} s_{23}^2 + a_{33} s_{31}^2 + a_{44} s_0^2 + \\
& 2 (a_{12} s_{12} s_{23} + a_{13} s_{12} s_{31} + a_{14} s_{12} s_0 + a_{23} s_{23} s_{31} + a_{24} s_{23} s_0 + a_{34} s_{31} s_0),
\end{align*}
where $\mathbf{s} = [s_{12}, s_{23}, s_{31}, s_0]^T$, and $M_\GG$'s elements are linear expressions of $\GG$'s parameters $K$, $L$, $A$, $B$, and $c$:
\begin{align*}
& a_{11} = -p_{xyz} - p_{yzx} + p_{zxy} - q_{xyz} - q_{yzx} + q_{zxy} + c \\
& a_{22} = +p_{xyz} - p_{yzx} - p_{zxy} + q_{xyz} - q_{yzx} - q_{zxy} + c \\
& a_{33} = -p_{xyz} + p_{yzx} - p_{zxy} - q_{xyz} + q_{yzx} - q_{zxy} + c \\
& a_{44} = +p_{xyz} + p_{yzx} + p_{zxy} + q_{xyz} + q_{yzx} + q_{zxy} + c \\
& a_{12} = +p_{xxy} + p_{zyz} + q_{xxy} + q_{zyz} \\
& a_{13} = +p_{yxy} + p_{zzx} + q_{yxy} + q_{zzx} \\
& a_{14} = +p_{xzx} - p_{yyz} - q_{xzx} + q_{yyz} \\
& a_{23} = +p_{xzx} + p_{yyz} + q_{xzx} + q_{yyz} \\
& a_{24} = +p_{yxy} - p_{zzx} - q_{yxy} + q_{zzx} \\
& a_{34} = -p_{xxy} + p_{zyz} + q_{xxy} - q_{zyz},
\end{align*}
where $p$ and $q$ are linear combinations:
\begin{align*}
p_{\alpha \beta \gamma} = (K_\alpha - L_\alpha)(A_\beta B_\gamma - A_\gamma B_\beta) \\
q_{\alpha \beta \gamma} = (A_\alpha - B_\alpha)(K_\beta L_\gamma - K_\gamma L_\beta).
\end{align*}
\end{theorem}
In this study, the following substitution is used
\begin{align*}
& P = K \times L \\
& Q = A - B \\
& U = K - L \\
& V = A \times B ,
\end{align*}
which leads to the corresponding matrix $M_\GG$
\begin{equation*}
\left|
\begin{smallmatrix}
    2 (P_3 Q_3 + U_3 V_3) - T + c & 2 (P_1 Q_3 + U_1 V_3) + R_2 & 2 (P_3 Q_2 + U_3 V_2) + R_1 & R_3 \\
    2 (P_1 Q_3 + U_1 V_3) + R_2 & 2 (P_1 Q_1 + U_1 V_1) - T + c & 2 (P_2 Q_1 + U_2 V_1) + R_3 & R_1 \\
    2 (P_3 Q_2 + U_3 V_2) + R_1 & 2 (P_2 Q_1 + U_2 V_1) + R_3 & 2 (P_2 Q_2 + U_2 V_2) - T + c & R_2 \\
    R_3 & R_1 & R_2 & T + c
\end{smallmatrix}
\right| ,
\end{equation*}
where $R$ and $T$ are defined as
\begin{align} \label{eqn_r_and_t}
R & = P \times Q + U \times V \\
T & = P \cdot Q + U \cdot V \notag .
\end{align}
The proposed theory can avoid using variables the $K$, $L$, $A$, and $B$ completely. Using the notations of $P$, $Q$, $U$, $V$, a general predicate $\GG_s = \GG_s(K, L, A, B, c)$ can have a much simpler form
\begin{equation}
\GG_s = P \cdot \Rot_s(Q) + U \cdot \Rot_s(V) + c .
\end{equation}
After changing to the new notation, the following obvious fact is essential in all reductions.
\begin{align}
P \cdot U & = 0 \label{eqn_p_dot_u_is_zero} \\
Q \cdot V & = 0 \label{eqn_q_dot_v_is_zero}
\end{align}
Following a common convention, a normalized vector $X \in \R^3$ is denoted as $\hat{X}$.
\begin{definition}
Let $X$ be a non-zero vector in $\R^3$. A unit-length vector $\hat{X}$
\begin{equation*}
\hat{X} = \frac{X}{\|X\|}
\end{equation*}
is called \emph{a normalized vector $X$}.
\end{definition}
In \cite{phd_dobrowolski}, the following general theorem was proved regarding spin-surfaces.
\begin{theorem}[Dobrowolski, \cite{phd_dobrowolski}] \label{thm_spin_surface_ellipsoid_hyperboloid}
A spin-quadric is equivalent to the intersection of a rotated three-dimensional ellipsoid or a hyperboloid centred at zero with a unit sphere in the ambient space $\R^4$.
\end{theorem}
This study increases our knowledge of the geometry of spin-quadrics, starting with theorem \ref{thm_spin_surface_ellipsoid_hyperboloid}. Further systematic studies of spin-surfaces are possible when the parameterization of a general predicate is available.

\section{Properties of a general predicate}

\subsection{Classification of general predicates}

In the remainder of this study, the following definition is used in order to omit degenerate predicates.
\begin{definition}[A proper general predicate]
A general predicate $\GG$ is called \textbf{proper} iff:
\begin{equation}
\|P\| \|Q\| \ne 0 \, \vee \, \|U\| \|V\| \ne 0 .
\end{equation}
A general predicate is called \textbf{improper} if it is not proper.
\end{definition}
We observe that when a predicate is improper $\|P\| \|Q\| = 0 \, \wedge \, \|U\| \|V\| = 0$, it is equal to:
\begin{equation*}
\GG_s = P \cdot \Rot_s(Q) + (U) \cdot \Rot_s(V) + c = c .
\end{equation*}
The allowed, forbidden, and border sets of rotations for such a predicate are empty or the whole space depending on the sign of $c$. In practical usages, improper predicates are not useful, and thus they are usually discarded.

Next, the following definition is introduced to distinguish two different types of a proper general predicate.
\begin{definition}[Toroidal and ellipsoidal general predicates]
Let $\GG$ be a \textbf{proper} general predicate. Iff
\begin{equation}
\|P\| \|Q\| = 0 \, \vee \, \|U\| \|V\| = 0,
\end{equation}
then $\GG$ is called \textbf{a toroidal general predicate}; otherwise, iff
\begin{equation}
\|P\| \|Q\| \ne 0 \, \wedge \, \|U\| \|V\| \ne 0,
\end{equation}
then $\GG$ is called \textbf{an ellipsoidal general predicate}.
\end{definition}
The motivation for the definitions given above is based on considerations given in one of the following sections. Note that every general predicate is either an improper or proper ellipsoidal, or a proper toroidal.

Next, we consider the eigenvalues and eigenvectors of $M_\GG$, which also help to determine the type of spin-surface and its exact parameterization.

\subsection{Calculation of eigenvalues}

The following lemma provides an insight into the spectrum of a matrix associated with a general predicate.

\begin{lemma} \label{lem_eigenvalues_of_general_predicate}
Let $\GG_s = P \cdot \Rot_s(Q) + U \cdot \Rot_s(V) + c$ be a general predicate. Assume that $S_\GG(s)$ is a spin-quadric representing $\GG_s$ and $M_\GG$ is the matrix of the associated quadratic form. The spectrum of $M_\GG$ comprises four real values $\lambda_{\alpha, \beta \in \{-1, 1\} }$, which are equal to:
\begin{equation*}
\lambda_{\alpha, \beta \in \{-1, 1\} } = c - \gamma_{\alpha, \beta},
\end{equation*}
where
\begin{equation*}
\gamma_{\alpha, \beta \in \{-1, 1\} } = \alpha \|P\| \|Q\| + \beta \|U\| \|V\| .
\end{equation*}
\end{lemma}
\begin{proof}
The spectrum of $M_\GG$ comprises eigenvalues $\lambda$ that satisfy $|M_\GG - \lambda I| = 0$. The matrix $M_\GG - \lambda I$ is equal to the following.
\begin{align*}
\left[
\begin{smallmatrix}
    2 (P_3 Q_3 + U_3 V_3) - T + c - \lambda & 2 (P_1 Q_3 + U_1 V_3) + R_2 & 2 (P_3 Q_2 + U_3 V_2) + R_1 & R_3 \\
    2 (P_1 Q_3 + U_1 V_3) + R_2 & 2 (P_1 Q_1 + U_1 V_1) - T + c - \lambda & 2 (P_2 Q_1 + U_2 V_1) + R_3 & R_1 \\
    2 (P_3 Q_2 + U_3 V_2) + R_1 & 2 (P_2 Q_1 + U_2 V_1) + R_3 & 2 (P_2 Q_2 + U_2 V_2) - T + c - \lambda & R_2 \\
    R_3 & R_1 & R_2 & T + c - \lambda
\end{smallmatrix}
\right]
\end{align*}
By setting $\gamma = c - \lambda$, the determinant becomes a function of $\gamma$
\begin{equation*}
h(\gamma) =
\left|
\begin{smallmatrix}
    2 (P_3 Q_3 + U_3 V_3) - T + \gamma & 2 (P_1 Q_3 + U_1 V_3) + R_2 & 2 (P_3 Q_2 + U_3 V_2) + R_1 & R_3 \\
    2 (P_1 Q_3 + U_1 V_3) + R_2 & 2 (P_1 Q_1 + U_1 V_1) - T + \gamma & 2 (P_2 Q_1 + U_2 V_1) + R_3 & R_1 \\
    2 (P_3 Q_2 + U_3 V_2) + R_1 & 2 (P_2 Q_1 + U_2 V_1) + R_3 & 2 (P_2 Q_2 + U_2 V_2) - T + \gamma & R_2 \\
    R_3 & R_1 & R_2 & T + \gamma
\end{smallmatrix}
\right| .
\end{equation*}
The determinant can be calculated and simplified effectively by using Laplace expansion with respect to the fourth row or column. Thus, we find that
\begin{align*}
h(\gamma) = \gamma^4 - 2 \left( \|P\|^2 \|Q\|^2 + \|U\|^2 \|V\|^2 \right) \gamma^2 + \left( \|P\|^2 \|Q\|^2 - \|U\|^2 \|V\|^2 \right)^2 = 0.
\end{align*}
Equation $h(\gamma)$ is bi-quadratic. If we let $\gamma^2 = \psi \ge 0$, then the equation becomes
\begin{equation*}
\psi^2  - 2 \left( \|P\|^2 \|Q\|^2 + \|U\|^2 \|V\|^2 \right) \psi + \left( \|P\|^2 \|Q\|^2 - \|U\|^2 \|V\|^2 \right)^2 = 0.
\end{equation*}
The quadratic discriminant is equal to
\begin{align*}
\Delta & = 16 \|P\|^2 \|Q\|^2 \|U\|^2 \|V\|^2 ,
\end{align*}
which implies that there are two real solutions, $\psi^-$ and $\psi^+$, and they are equal to
\begin{align*}
\psi^- & = ( \|P\| \|Q\| - \|U\| \|V\| )^2 \\
\psi^+ & = ( \|P\| \|Q\| + \|U\| \|V\| )^2 .
\end{align*}
The four corresponding $\gamma = \pm \sqrt{\psi}$ values are
\begin{equation*}
\gamma_{i=1..4} = \pm \|P\| \|Q\| \pm \|U\| \|V\| ,
\end{equation*}
which correspond to the four eigenvalues $\lambda = c -\gamma$ of matrix $M_\GG$
\begin{equation*}
\lambda_{i=1..4} = c - \left( \pm \|P\| \|Q\| \pm \|U\| \|V\| \right),
\end{equation*}
as it was shown.
\end{proof}

The following two lemmas characterize the spectrum of $M_\GG$ depending on the type of a general predicate.

\begin{lemma} \label{lem_eigenvalues_of_ellipsoidal_general_predicate}
Let $\GG_s$ be an ellipsoidal general predicate. Assume that $S_\GG(s)$ is a spin-quadric representing $\GG_s$ and $M_\GG$ is the matrix of the associated quadratic form. Then, the spectrum of $M_\GG$ comprises four different real values.
\end{lemma}
\begin{proof}
Due to Lemma \ref{lem_eigenvalues_of_general_predicate}, $M_\GG$ only contains real values. By contradiction, suppose that the spectrum of $M_\GG$ contains at least one pair of equal values $\lambda_1$, $\lambda_2$ such that $\lambda_1 = \lambda_2$. There are essentially three different cases and all other cases follow from these by swapping $\lambda_1$ and $\lambda_2$.
\begin{enumerate}
\item If $\lambda_1 = c - \left( + \|P\| \|Q\| + \|U\| \|V\| \right)$ and $\lambda_2 = c - \left( + \|P\| \|Q\| - \|U\| \|V\| \right)$, \\
      then $\|U\| \|V\| = 0$, which is a contradiction since $\GG_s$ is an ellipsoidal general predicate.
\item If $\lambda_1 = c - \left( + \|P\| \|Q\| + \|U\| \|V\| \right)$ and $\lambda_2 = c - \left( - \|P\| \|Q\| + \|U\| \|V\| \right)$, \\
      then $\|P\| \|Q\| = 0$, which is a contradiction since $\GG_s$ is an ellipsoidal general predicate.
\item If $\lambda_1 = c - \left( + \|P\| \|Q\| + \|U\| \|V\| \right)$ and $\lambda_2 = c - \left( - \|P\| \|Q\| - \|U\| \|V\| \right)$, \\
      then $\|P\| \|Q\| + \|U\| \|V\| = 0$, which is a contradiction since $\GG_s$ is a proper predicate.
\end{enumerate}
In all cases, there is a contradiction so the lemma must be true.
\end{proof}

\begin{lemma} \label{lem_eigenvalues_of_toroidal_general_predicate}
Let $\GG_s$ be a toroidal general predicate. Assume that $S_\GG(s)$ is a spin-quadric representing $\GG_s$ and $M_\GG$ is the matrix of the associated quadratic form. Then, the spectrum of $M_\GG$ comprises two pairs of different real values.
\end{lemma}
\begin{proof}
Due to Lemma \ref{lem_eigenvalues_of_general_predicate}, $M_\GG$ only contains real values. Without loss of generality, it can be assumed that $\|P\| \|Q\| = 0$ (the other case is analogous). From Lemma \ref{lem_eigenvalues_of_general_predicate}, it follows that the eigenvalues are as follows.
\begin{align*}
\lambda_1 & = c - \left(+ \|P\| \|Q\| + \|U\| \|V\| \right) = c - \|U\| \|V\| \\
\lambda_2 & = c - \left(+ \|P\| \|Q\| - \|U\| \|V\| \right) = c + \|U\| \|V\| \\
\lambda_3 & = c - \left(-  \|P\| \|Q\| + \|U\| \|V\| \right) = c - \|U\| \|V\| \\
\lambda_4 & = c - \left(- \|P\| \|Q\| - \|U\| \|V\| \right) = c + \|U\| \|V\|
\end{align*}
It is obvious that $\lambda_1 = \lambda_3$ and $\lambda_2 = \lambda_4$. Moreover, it is impossible for $\lambda_1 = \lambda_2$ because this would imply that $\|U\| \|V\| = 0$ and that $\GG_s$ is improper, which contradicts the assumptions. Thus, the lemma is proved.
\end{proof}

The spectrum of the matrix associated with a proper general predicate can contain zeroes. However, the spectrum must contain at least one non-zero eigenvalue due to the following lemma.
\begin{lemma} \label{lem_non_zero_spectrum_of_general_predicate}
Let $\GG_s$ be a proper general predicate. Assume that $S_\GG(s)$ is a spin-quadric representing $\GG_s$ and $M_\GG$ is the matrix of the associated quadratic form. Then, the spectrum of $M_\GG$ contains at least one non-zero eigenvalue.
\end{lemma}
\begin{proof}
By contradiction, suppose that all the eigenvalues are zero. Then, it follows that
\begin{align*}
\lambda_1 & = c - \left(+ \|P\| \|Q\| + \|U\| \|V\| \right) = 0 \\
\lambda_2 & = c - \left(+ \|P\| \|Q\| - \|U\| \|V\| \right) = 0 \\
\lambda_3 & = c - \left(-  \|P\| \|Q\| + \|U\| \|V\| \right) = 0 \\
\lambda_4 & = c - \left(- \|P\| \|Q\| - \|U\| \|V\| \right) = 0,
\end{align*}
which can be reduced to
\begin{equation*}
\|P\| \|Q\| = 0 \, \wedge \, \|U\| \|V\| = 0 \, \wedge \, c = 0,
\end{equation*}
and this contradicts the assumption that the predicate is proper.
\end{proof}
We can conclude by combining Lemmas \ref{lem_eigenvalues_of_toroidal_general_predicate} and \ref{lem_eigenvalues_of_ellipsoidal_general_predicate} with \ref{lem_non_zero_spectrum_of_general_predicate}.
\begin{corollary}
The following two statements characterize the zeroes in the spectrum of a matrix associated with a proper general predicate.
\begin{enumerate}   
\item If the spectrum of the matrix associated with a toroidal general predicate contains zeroes, then it contains exactly two zeroes.
\item If the spectrum of the matrix associated with a toroidal general predicate contains zero, then there is only one zero in the spectrum.
\end{enumerate}
\end{corollary}

An improper general predicate can also be characterized by its spectrum.
\begin{lemma} \label{lem_eigenvalues_of_improper_general_predicate}
Let $\GG_s$ be a general predicate. Assume that $S_\GG(s)$ is a spin-quadric representing $\GG_s$ and $M_\GG$ is the matrix of the associated quadratic form. The predicate $\GG_s$ is improper iff the spectrum of $M_\GG$ comprises four equal real values.
\end{lemma}
\begin{proof}
Due to Lemma \ref{lem_eigenvalues_of_general_predicate}, $M_\GG$ only contains real values. If we consider the improper $\GG_s$ predicate, both $\|P\| \|Q\|$ and $\|U\| \|V\|$ must be zero. From Lemma \ref{lem_eigenvalues_of_general_predicate}, it follows that the eigenvalues
\begin{align*}
\lambda_1 & = c - \left(+ \|P\| \|Q\| + \|U\| \|V\| \right) = c \\
\lambda_2 & = c - \left(+ \|P\| \|Q\| - \|U\| \|V\| \right) = c \\
\lambda_3 & = c - \left(- \|P\| \|Q\| + \|U\| \|V\| \right) = c \\
\lambda_4 & = c - \left(- \|P\| \|Q\| - \|U\| \|V\| \right) = c,
\end{align*}
are all equal, as we showed. If we consider the opposite case, then we suppose that $\lambda_i = x$ is the spectrum of a general predicate. From Lemma \ref{lem_eigenvalues_of_general_predicate}, it follows that the eigenvalues must satisfy the following set of equations
\begin{align*}
x = \lambda_1 & = c - \left(+ \|P\| \|Q\| + \|U\| \|V\| \right) \\
x = \lambda_2 & = c - \left(+ \|P\| \|Q\| - \|U\| \|V\| \right) \\
x = \lambda_3 & = c - \left(-  \|P\| \|Q\| + \|U\| \|V\| \right) \\
x = \lambda_4 & = c - \left(- \|P\| \|Q\| - \|U\| \|V\| \right),
\end{align*}
which reduce to $\|P\| \|Q\| = 0  \wedge \|U\| \|V\| = 0 \wedge c = x$. For these parameters, a general predicate is improper, as shown.
\end{proof}

\subsection{Calculation of eigenvectors}

The symbolic computation of the eigenvectors of $M_\GG$ is much more involved and it requires a number of special cases, including different cases for a toroidal and ellipsoidal general predicates.

\begin{lemma} \label{lem_eigenvectors_of_general_predicate}
Let $\GG_s = P \cdot \Rot_s(Q) + U \cdot \Rot_s(V) + c$ be a general predicate. Assume that $S_\GG(s)$ is a spin-quadric representing $\GG_s$ and $M_\GG$ is the matrix of the associated quadratic form. For each eigenvalue $\lambda_{\alpha, \beta \in \{-1, 1\} } = c - \gamma_{\alpha, \beta}$ of $M_\GG$, the corresponding eigenvector $W = [\mathbf{w}_{12}, \mathbf{w}_{23}, \mathbf{w}_{31}, \mathbf{w}_0]^T$ depends on the type of $\GG_s$, as follows.
\begin{enumerate}
\item Case $\GG_s$ is an ellipsoidal general predicate: the coordinates of the eigenvector $W$ are related to the coordinates of \emph{pinor} $\mathbf{w}$
\begin{equation*}
\mathbf{w} = 1 - \alpha \beta \, \hat{P} \hat{U} \hat{Q} \hat{V} - \alpha \, \hat{P} \hat{Q} - \beta \, \hat{U} \hat{V}
\end{equation*}
where $\mathbf{w} = \mathbf{w}_{12} \e_{12} + \mathbf{w}_{23} \e_{23} + \mathbf{w}_{31} \e_{31} + \mathbf{w}_0 \in \Pin(3)$. The eigenvector $W$ is valid in all cases where $\mathbf{w}_0 \ne 0$.
\item Case $\GG_s$ is a toroidal general predicate: a pair of equal eigenvalues corresponds to a two-dimensional eigenplane. Each of the three following planes is an eigenplane if both of the vectors that form the plane are non-zero.
\begin{align*}
Z^{(1)} = & \left( \begin{smallmatrix}
        \hat{P}_2 \hat{P}_1 - \hat{Q}_2 \hat{Q}_1 - \alpha (\hat{P}_2 \hat{Q}_1 - \hat{P}_1 \hat{Q}_2) \\
        0 \\
        - \hat{P}_3 \hat{P}_1 + \hat{Q}_3 \hat{Q}_1 + \alpha (\hat{P}_3 \hat{Q}_1 - \hat{P}_1 \hat{Q}_3) \\
        - \hat{P}_1 \hat{P}_1 - \hat{Q}_1 \hat{Q}_1 + \alpha (\hat{P}_1 \hat{Q}_1 + \hat{P}_1 \hat{Q}_1)
    \end{smallmatrix} \right) u
    + \left( \begin{smallmatrix}
        0 \\
        \hat{P}_2 \hat{P}_1 - \hat{Q}_2 \hat{Q}_1 - \alpha (\hat{P}_2 \hat{Q}_1 - \hat{P}_1 \hat{Q}_2) \\
        - \hat{P}_1 \hat{P}_1 + \hat{Q}_1 \hat{Q}_1 + \alpha (\hat{P}_1 \hat{Q}_1 - \hat{P}_1 \hat{Q}_1) \\
        \hat{P}_3 \hat{P}_1 + \hat{Q}_3 \hat{Q}_1 - \alpha (\hat{P}_3 \hat{Q}_1 + \hat{P}_1 \hat{Q}_3)
    \end{smallmatrix} \right) v
\end{align*}
or
\begin{align*}
Z^{(2)} = & \left( \begin{smallmatrix}
        \hat{P}_2 \hat{P}_2 - \hat{Q}_2 \hat{Q}_2 - \alpha (\hat{P}_2 \hat{Q}_2 - \hat{P}_2 \hat{Q}_2) \\
        0 \\
        - \hat{P}_3 \hat{P}_2 + \hat{Q}_3 \hat{Q}_2 + \alpha (\hat{P}_3 \hat{Q}_2 - \hat{P}_2 \hat{Q}_3) \\
        - \hat{P}_1 \hat{P}_2 - \hat{Q}_1 \hat{Q}_2 + \alpha (\hat{P}_1 \hat{Q}_2 + \hat{P}_2 \hat{Q}_1)
    \end{smallmatrix} \right) u
    + \left( \begin{smallmatrix}
        0 \\
        \hat{P}_2 \hat{P}_2 - \hat{Q}_2 \hat{Q}_2 - \alpha (\hat{P}_2 \hat{Q}_2 - \hat{P}_2 \hat{Q}_2) \\
        - \hat{P}_1 \hat{P}_2 + \hat{Q}_1 \hat{Q}_2 + \alpha (\hat{P}_1 \hat{Q}_2 - \hat{P}_2 \hat{Q}_1) \\
        \hat{P}_3 \hat{P}_2 + \hat{Q}_3 \hat{Q}_2 - \alpha (\hat{P}_3 \hat{Q}_2 + \hat{P}_2 \hat{Q}_3)
    \end{smallmatrix} \right) v
\end{align*}
or
\begin{align*}
Z^{(3)} = & \left( \begin{smallmatrix}
        \hat{P}_2 \hat{P}_3 - \hat{Q}_2 \hat{Q}_3 - \alpha (\hat{P}_2 \hat{Q}_3 - \hat{P}_3 \hat{Q}_2) \\
        0 \\
        - \hat{P}_3 \hat{P}_3 + \hat{Q}_3 \hat{Q}_3 + \alpha (\hat{P}_3 \hat{Q}_3 - \hat{P}_3 \hat{Q}_3) \\
        - \hat{P}_1 \hat{P}_3 - \hat{Q}_1 \hat{Q}_3 + \alpha (\hat{P}_1 \hat{Q}_3 + \hat{P}_3 \hat{Q}_1)
    \end{smallmatrix} \right) u
    + \left( \begin{smallmatrix}
        0 \\
        \hat{P}_2 \hat{P}_3 - \hat{Q}_2 \hat{Q}_3 - \alpha (\hat{P}_2 \hat{Q}_3 - \hat{P}_3 \hat{Q}_2) \\
        - \hat{P}_1 \hat{P}_3 + \hat{Q}_1 \hat{Q}_3 + \alpha (\hat{P}_1 \hat{Q}_3 - \hat{P}_3 \hat{Q}_1) \\
        \hat{P}_3 \hat{P}_3 + \hat{Q}_3 \hat{Q}_3 - \alpha (\hat{P}_3 \hat{Q}_3 + \hat{P}_3 \hat{Q}_3)
    \end{smallmatrix} \right) v
\end{align*}
\end{enumerate}
\end{lemma}
\begin{proof}
A long proof is provided in the appendix. All of the special cases that are not listed in the lemma are handled specifically in the proof.
\end{proof}

\subsection{Parameterization of a general predicate}

Let $\GG$ be a proper general predicate. Assume that $S_\GG(s)$ is a spin-quadric representing $\GG$ and $M_\GG$ is the matrix of the associated quadratic form. By a parameterization of a general predicate, we mean the parameterization of the spin-quadric $S_\GG(s)$. A parameterization is derived by an eigenvalue decomposition of the matrix $M_\GG$. According to Lemmas \ref{lem_eigenvalues_of_general_predicate} and \ref{lem_eigenvectors_of_general_predicate}, a decomposition of $M_\GG$ exists. If we denote the eigenvalues of $M_\GG$ as
\begin{align*}
\lambda_1 & = c - \gamma_1 = c - \gamma_{+1, +1} = c - \left(+ \|P\| \|Q\| + \|U\| \|V\| \right) \\
\lambda_2 & = c - \gamma_2 = c - \gamma_{+1, -1} = c - \left(+ \|P\| \|Q\| - \|U\| \|V\| \right) \\
\lambda_3 & = c - \gamma_3 = c - \gamma_{-1, +1} = c - \left(- \|P\| \|Q\| + \|U\| \|V\| \right) \\
\lambda_4 & = c - \gamma_4 = c - \gamma_{-1, -1} = c - \left(- \|P\| \|Q\| - \|U\| \|V\| \right), \\
\end{align*}
then the decomposition is
\begin{equation} \label{parameterization_spectral_decomposition}
M_\GG = Q \Lambda Q^T ,
\end{equation}
where matrix $Q$ is a matrix of normalized eigenvectors $W$ and $\Lambda$ is a diagonal matrix of the corresponding eigenvalues
\begin{equation*}
Q =
\left[ \begin{smallmatrix}
    \frac{W^{(1)}_1}{\|W^{(1)}\|} & \frac{W^{(2)}_1}{\|W^{(2)}\|} & \frac{W^{(3)}_1}{\|W^{(3)}\|} & \frac{W^{(4)}_1}{\|W^{(4)}\|} \\
    \frac{W^{(1)}_2}{\|W^{(1)}\|} & \frac{W^{(2)}_2}{\|W^{(2)}\|} & \frac{W^{(3)}_2}{\|W^{(3)}\|} & \frac{W^{(4)}_2}{\|W^{(4)}\|} \\
    \frac{W^{(1)}_3}{\|W^{(1)}\|} & \frac{W^{(2)}_3}{\|W^{(2)}\|} & \frac{W^{(3)}_3}{\|W^{(3)}\|} & \frac{W^{(4)}_3}{\|W^{(4)}\|} \\
    \frac{W^{(1)}_4}{\|W^{(1)}\|} & \frac{W^{(2)}_4}{\|W^{(2)}\|} & \frac{W^{(3)}_4}{\|W^{(3)}\|} & \frac{W^{(4)}_4}{\|W^{(4)}\|}
\end{smallmatrix} \right],
\quad \Lambda =
\left[ \begin{smallmatrix}
    \lambda_1 & 0 & 0 & 0 \\
    0 & \lambda_2 & 0 & 0 \\
    0 & 0 & \lambda_3 & 0 \\
    0 & 0 & 0 & \lambda_4
\end{smallmatrix} \right],
\end{equation*}
where the notation $W^{(j)}_i$ denotes the $i$-th coordinate of the eigenvector corresponding to eigenvalue $\lambda_j$.
Matrix $M_\GG$ is a real symmetric matrix so $Q$ is a rotation matrix comprising four orthonormal eigenvectors.
The spin-surface is determined by the set of zeroes of
\begin{equation}
S_\GG(s) = \mathbf{s}^T M_\GG \mathbf{s} = 0
\end{equation}
and by plugging in the decomposition \eqref{parameterization_spectral_decomposition}, we obtain
\begin{equation*}
S_\GG(s) = \mathbf{s}^T (Q \Lambda Q^T) \mathbf{s} = (\mathbf{s}^T Q) \Lambda (Q^T \mathbf{s}).
\end{equation*}
By setting
\begin{equation*}
\mathbf{t} = Q^T \mathbf{s},
\end{equation*}
we have
\begin{equation*}
\mathbf{s}^T Q = \mathbf{t}^T, \quad \mathbf{s} = Q \mathbf{t}, \quad \mathbf{s}^T = \mathbf{t}^T Q^T
\end{equation*}
and the decomposition becomes
\begin{equation}
S_\GG(s) = \mathbf{t}^T \Lambda \mathbf{t} = 0.
\end{equation}
Each $\mathbf{s}$ satisfies a normalization relation
\begin{equation}
\mathbf{s}^T \mathbf{s} = 1,
\end{equation}
which after substituting with $\mathbf{t}$ becomes
\begin{align*}
& (\mathbf{t}^T Q^T) (Q \mathbf{t}) = 1 \\
& \mathbf{t}^T \mathbf{t} = 1.
\end{align*}
If we let $\mathbf{t} = [ t_{12}, t_{23}, t_{31}, t_0 ]^T$, then $\mathbf{t}^T \Lambda \mathbf{t}$ is equal to
\begin{align} \label{eqn_parameterization_cases}
\begin{cases}
\lambda_1 t_{12}^2 + \lambda_2 t_{23}^2 + \lambda_3 t_{31}^2 + \lambda_4 t_0^2 = 0 \\
t_{12}^2 + t_{23}^2 + t_{31}^2 + t_0^2 = 1.
\end{cases}
\end{align}
By substituting $\lambda$ in \eqref{eqn_parameterization_cases}, we obtain
\begin{align} \label{eqn_parameterization_cases_substituted}
\begin{cases}
\gamma_1 t_{12}^2 + \gamma_2 t_{23}^2 + \gamma_3 t_{31}^2 + \gamma_4 t_0^2 = c \\
t_{12}^2 + t_{23}^2 + t_{31}^2 + t_0^2 = 1.
\end{cases}
\end{align}

The parameterization will be different depending on the values in the spectrum of the matrix associated with a spin-quadric. Denote $a = \|P\|\|Q\|$ and $b = \|U\|\|V\|$.

\subsubsection{Parameterization of an ellipsoidal general predicate}

For an ellipsoidal general predicate, $a \ne 0$ and $b \ne 0$. Equations \eqref{eqn_parameterization_cases_substituted} become
\begin{align} \label{eqn_parameterization_cases_ellipsoidal}
\begin{cases}
(a + b) t_{12}^2 + (a - b) t_{23}^2 + (-a + b) t_{31}^2 + (-a - b) t_0^2 = c \\
t_{12}^2 + t_{23}^2 + t_{31}^2 + t_0^2 = 1.
\end{cases}
\end{align}

Lemma \ref{lem_solutions_ellipsoidal} provides a solution to the set of equations \eqref{eqn_parameterization_cases_ellipsoidal}.

\begin{lemma} \label{lem_solutions_ellipsoidal}
A solution to the following set of equations involving $x, y, z, w$
\begin{align*}
\begin{cases}
& (a + b) x^2 + (a - b) y^2 + (-a + b) z^2 + (-a - b) w^2 = c \\
& x^2 + y^2 + z^2 + w^2 = 1
\end{cases}
\end{align*}
for the parameters $a \in \R^+, b \in \R^+, c \in \R$ can be parameterized according to the formulae given in the proof.
\end{lemma}
\begin{proof}
We begin by substituting $w^2 = 1 - x^2 - y^2 - z^2$ from the second equation into the first equation
\begin{align}  \label{eqn_sols_ell_subst}
\begin{cases}
& 2 (a + b) x^2 + 2 a y^2 + 2 b z^2 = a + b + c \\
& w^2 = 1 - x^2 - y^2 - z^2 .
\end{cases}
\end{align}
We observe that when $a + b + c < 0 \iff c < - a - b$, there are no real solutions at all:
\begin{itemize}
    \item sub-case $c < - a - b$ labelled \emph{"an empty set"} \\
          no real solutions;
\end{itemize}
For $a + b + c = 0 \iff c = - a - b$, there is a pair of solutions
\begin{equation*}
x = 0, \quad y = 0, \quad z = 0, \quad w = \pm 1,
\end{equation*}
which are covered by the case:
\begin{itemize}
    \item sub-case $c = - a - b$ labelled \emph{"a pair of points"} \\
          empty domain hole (explained later);
\end{itemize}
otherwise, we have $a + b + c > 0 \iff c > - a - b$ and further modifications can be made. We can rewrite the first equation in the form of an ellipsoid
\begin{equation*}
\left(\frac{x}{\sqrt{\frac{a + b + c}{2 (a + b)}}}\right)^2 + \left(\frac{y}{\sqrt{\frac{a + b + c}{2 a}}}\right)^2 + \left(\frac{z}{\sqrt{\frac{a + b + c}{2 b}}}\right)^2 = 1.
\end{equation*}

The ellipsoid radii are $r_x = \sqrt{\frac{a + b + c}{2 (a + b)}}$, $r_y = \sqrt{\frac{a + b + c}{2 a}}$, and $r_z = \sqrt{\frac{a + b + c}{2 b}}$. It should be noted that it is always the case that $r_x < r_y$ and $r_x < r_z$. Although an ellipsoid can be parameterized easily, it has to be parameterized so the domain satisfies the second equation with $w = \pm \sqrt{1 - x^2 - y^2 - z^2}$. The value under the root sign must be non-negative, or the domain would be too wide and some solutions would be imaginary.
Geometrical interpretation can help to find the appropriate cases. A valid triple $x$, $y$, and $z$ is a point on the ellipsoid such that the distance to the origin is not greater than $1$. Thus, the points on the ellipsoid located inside a three-dimensional closed unit ball are valid. The two parameterizations for positive $w$ and negative $w$ are simply called a positive parameterization and a negative parameterization. When a positive parameterization on the ellipsoid reaches a unit ball, or equivalently when $w = 0$, it extends to the negative parameterization, which is a mirrored copy of the first domain along the contact curve with the unit ball. The set where $w = 0$ is called \emph{a domain hole}. In each case, when a domain hole is non-empty, a process of \emph{domain sewing} is required along the boundary of a domain hole in order to handle the connectivity correctly between different parts of the domain. Three main cases depend on the radii lengths, as follows.

\begin{enumerate}
    \item Case $r_x \le 1 \wedge r_y \le 1 \wedge r_z \le 1$: all three radii are not greater than $1$ \\
        Excluding cases when the radius is equal to $1$, the whole ellipsoid fits inside the unit ball and the parameterization covers all of the points on the ellipsoid.
        \begin{itemize}
            \item Sub-case $r_x < 1 \wedge r_y < 1 \wedge r_z < 1$ labelled \emph{"a pair of separate ellipsoids"} \\
                  equivalent condition: $c < \min(a - b, b - a)$, \\
                  empty domain hole.
        \end{itemize}
        There are four boundary sub-cases with one or more radii equal to $1$.
        \begin{itemize}
            \item Sub-case $r_x = 1 \wedge r_y \le 1 \wedge r_z \le 1$ labelled \emph{"an empty set"} \\
                  equivalent condition: $c = a + b \wedge a \le 0 \wedge b \le 0$ \\
                  contradicts the assumptions of the lemma
            \item Sub-case $r_x < 1 \wedge r_y = 1 \wedge r_z < 1$ labelled \emph{"a pair of $y$-touching ellipsoids"} \\
                  equivalent condition: $c = a - b \wedge a < b$ \\
                  point domain hole
            \item sub-case $r_x < 1 \wedge r_y < 1 \wedge r_z = 1$ labelled \emph{"a pair of $z$-touching ellipsoids"} \\
                  equivalent condition: $c = b - a \wedge a > b$ \\
                  point domain hole
            \item Sub-case $r_x < 1 \wedge r_y = 1 \wedge r_z = 1$ labelled \emph{"a pair of $yz$-crossed ellipsoids"} \\
                  equivalent condition: $0 = c < a + b \wedge a = b$ \\
                  circular domain hole
        \end{itemize}

    \item Case $r_x > 1 \wedge r_y > 1 \wedge r_z > 1$: all three radii are greater than $1$ \\
        In this case, the whole ellipsoid is outside the unit ball and the parameterization domain is empty.
        \begin{itemize}
            \item Sub-case $r_x > 1 \wedge r_y > 1 \wedge r_z > 1$ labelled \emph{"an empty set"} \\
                  equivalent condition: $c > a + b$ \\
                  empty domain hole
        \end{itemize}

    \item Case $(r_x \le 1 \vee r_y \le 1 \vee r_z \le 1) \wedge (r_x > 1 \vee r_y > 1 \vee r_z > 1)$: at least one radius is not greater than $1$ and at least one radius is greater than $1$ \\
        We observe that $r_x$ is always strictly the shortest radius. Thus, $r_x$ must not be greater than $1$, otherwise it would be greater
        than $1$ and the remaining two radii would also be greater than $1$, which implies that all the radii are greater than $1$, as covered
        by the previous case.
        \begin{itemize}
            \item Sub-case $r_x > 1$ labelled \emph{"an empty set"} \\
                  equivalent condition: $c > a + b$ \\
                  contradicts the assumptions of the lemma
        \end{itemize}
        Thus, if $r_x$ is not greater than $1$, then either $r_y$ is greater than $1$ or $r_z$ is greater than $1$, or both $r_y$ and $r_z$ are greater than $1$.
        For the first two cases, the solution set is similar to a barrel (an ellipsoid with two opposite caps removed) and the remaining case represents a solution made from two ellipsoid
        caps, which were removed in the previous cases. \\
        \\
        \textbf{$y$-barrel-like} sub-cases $r_x \le 1 \wedge r_y > 1 \wedge r_z \le 1$

        \begin{itemize}
            \item Sub-case $r_x < 1 \wedge r_y > 1 \wedge r_z < 1$ labelled \emph{"a $y$-barrel"} \\
                  equivalent condition: $c \in \mathopen{}\mathclose{\left(a - b, b - a \right)}$ \\
                  deformed ellipse domain hole
        \end{itemize}
        There are three boundary sub-cases with one or more radii equal to $1$.
        \begin{itemize}
            \item Sub-case $r_x = 1 \wedge r_y > 1 \wedge r_z < 1$ labelled \emph{"an empty set"} \\
                  equivalent condition: $a < 0$ \\
                  contradicts the assumptions of the lemma
            \item Sub-case $r_x < 1 \wedge r_y > 1 \wedge r_z = 1$ labelled \emph{"a notched $y$-barrel"} \\
                  equivalent condition: $c = b - a > 0$ \\
                  deformed ellipse domain hole
            \item Sub-case $r_x = 1 \wedge r_y > 1 \wedge r_z = 1$ labelled \emph{"an empty set"} \\
                  equivalent condition: $a = 0$ \\
                  contradicts the assumptions of the lemma
        \end{itemize}
        \hfill
        \\
        \textbf{$z$-barrel-like} sub-cases $r_x \le 1 \wedge r_y \le 1 \wedge r_z > 1$

        \begin{itemize}
            \item Sub-case $r_x < 1 \wedge r_y < 1 \wedge r_z > 1$ labelled \emph{"a $z$-barrel"} \\
                  equivalent condition: $c \in \mathopen{}\mathclose{\left(b - a, a - b \right)}$ \\
                  deformed ellipse domain hole
        \end{itemize}
        There are three boundary sub-cases with one or more radii equal to $1$.
        \begin{itemize}
            \item Sub-case $r_x = 1 \wedge r_y < 1 \wedge r_z > 1$ labelled \emph{"an empty set"}
                  equivalent condition: $b < 0$ \\
                  contradicts the assumptions of the lemma
            \item Sub-case $r_x < 1 \wedge r_y = 1 \wedge r_z > 1$ labelled \emph{"a notched $z$-barrel"}
                  equivalent condition: $c = a - b > 0$ \\
                  deformed ellipse domain hole
            \item Sub-case $r_x = 1 \wedge r_y = 1 \wedge r_z > 1$ labelled \emph{"an empty set"}
                  equivalent condition: $b = 0$ \\
                  contradicts the assumptions of the lemma
        \end{itemize}
        \hfill
        \\
        \textbf{yz-caps-like} sub-cases $r_x \le 1 \wedge r_y > 1 \wedge r_z > 1$

        \begin{itemize}
            \item Sub-case $r_x < 1 \wedge r_y > 1 \wedge r_z > 1$ labelled \emph{"a pair of separate yz-caps"} \\
                  equivalent condition: $c \in \mathopen{}\mathclose{\left(\max(a - b, b - a), a + b \right)}$ \\
                  deformed ellipse domain hole
        \end{itemize}
        There is one boundary sub-case with one or more radii equal to $1$.
        \begin{itemize}
            \item Sub-case $r_x = 1 \wedge r_y > 1 \wedge r_z > 1$ labelled \emph{"a pair of separate yz-caps"} \\
                  equivalent condition: $c = a + b$ \\
                  deformed ellipse domain hole
      \end{itemize}
\end{enumerate}
Depending on the values of $a$ and $b$, there are $11$ different cases for parameterization, as summarized in Table \ref{ell_param_domain}.
\begin{table}[ht]
\begin{center}
\begin{tabular}{|c|l|l|l|}
\hline
param. type & $a < b$ & $a = b = t / 2$ & $a > b$ \\ \hline
\textbf{1/1/1} & $c \in \mathopen{}\mathclose{\left(-\infty, - a - b \right)}$ & $c \in \mathopen{}\mathclose{\left(-\infty, - t \right)}$ & $c \in \mathopen{}\mathclose{\left(-\infty, - a - b \right)}$ \\ \hline
\textbf{2/2/2} & $c = - a - b$ & $c = - t$ & $c = - a - b$ \\ \hline
\textbf{3/3/3} & $c \in \mathopen{}\mathclose{\left(- a - b, a - b \right)}$ & $c \in \mathopen{}\mathclose{\left(- t, 0 \right)}$ & $c \in \mathopen{}\mathclose{\left(- a - b, b - a \right)}$ \\ \hline
\textbf{4/5/6} & $c = a - b$ & $c = 0$ & $c = b - a$ \\ \hline
\textbf{7/1/8} & $c \in \mathopen{}\mathclose{\left(a - b, b - a \right)}$ & $ c \in \emptyset$ & $c \in \mathopen{}\mathclose{\left(b - a, a - b \right)}$ \\ \hline
\textbf{9/1/10} & $c = b - a$ & $ c \in \emptyset$ & $c = a - b$ \\ \hline
\textbf{11/11/11} & $c \in \mathopen{}\mathclose{\left(b - a, a + b \right)}$ & $c \in \mathopen{}\mathclose{\left(0, t \right)}$ & $c \in \mathopen{}\mathclose{\left(a - b, a + b \right)}$ \\ \hline
\textbf{11/11/11} & $c = a + b$ & $c = t$ & $c = a + b$ \\ \hline
\textbf{1/1/1} & $c \in \mathopen{}\mathclose{\left(a + b, + \infty \right)}$ & $c \in \mathopen{}\mathclose{\left(t, + \infty \right)}$ & $c \in \mathopen{}\mathclose{\left(a + b, + \infty \right)}$ \\ \hline
\end{tabular}
\end{center}
\caption{Parameterization types for different values of $a$ and $b$}
\label{ell_param_domain}
\end{table}
These cases are numbered according to the following list of \emph{parameterization types}.
\begin{enumerate}
\item[(1)] an empty case
\item[(2)] a pair of points
\item[(3)] a pair of separate ellipsoids
\item[(4)] a pair of $y$-touching ellipsoids
\item[(5)] a pair of $yz$-crossed ellipsoids
\item[(6)] a pair of $z$-touching ellipsoids
\item[(7)] a $y$-barrel
\item[(8)] a $z$-barrel
\item[(9)] a notched $y$-barrel
\item[(10)] a notched $z$-barrel
\item[(11)] a pair of separate $yz$-caps
\end{enumerate}

Each case is discussed separately in the following.

\paragraph{Cases 1}
There is no parameterization. The number of connected components is equal to zero. The domain hole is empty. The dimension of the solution is $-1$.

\paragraph{Case 2}
The parameterization is trivial and it comprises two points: $(x, y, z, w) = (0, 0, 0, \pm 1)$. There are two connected components. The domain hole is empty. The dimension of the solution is $0$.

\paragraph{Case 3}
In this case, the solution is a pair of separate ellipsoids. Using a common trigonometric parameterization of an ellipsoid, we write:
\begin{align*}
\begin{cases}
& x = \sqrt{\frac{a + b + c}{2 (a + b)}} \sin(\beta) \cos(\alpha) \\
& y = \sqrt{\frac{a + b + c}{2 a}} \sin(\beta) \sin(\alpha) \\
& z = \sqrt{\frac{a + b + c}{2 b}} \cos(\beta) \\
& w = \sigma \sqrt{1 - x^2 - y^2 - z^2}
\end{cases}
\end{align*}
for $\alpha \in [0, 2 \pi)$, $\beta \in [0, \pi]$, and $\sigma \in \{ -1, 1\}$. In this case, $w = 0$ is impossible so there are two separate ellipsoids as well as two components. The domain hole is empty. The dimension of the solution is $2$.

\paragraph{Cases 4, 5, and 6}
In these cases, $r_x = 1$ is impossible because this implies that $c = a + b$, which is not covered by this case. There are always two solutions for $w$. Note that for $r_y < 1$ and $r_z < 1$, it is always the case that $w \ne 0$ so the two solutions never overlap and there are two connected components (a pair of separate ellipsoids, which are covered by the previous case).

In the other case, at least one of $r_y$, $r_z$ is equal to $1$ and the parameters $(\alpha, \beta)$ exist where $w = 0$, and the two solutions have a common point. This occurs when $c = a - b$ or $c = b - a$, or equivalently when $c = \min(a - b, b - a)$. The subset of $w = 0$ is zero-dimensional (when only one radius is equal to $1$) or one-dimensional (when two radii are equal to $1$; equivalently, $a = b \wedge c = 0$). The two zero-dimensional sub-cases are called $y$-touching or $z$-touching ellipsoids according to the axis in the direction where the ellipsoids intersect. The other one-dimensional case is called $yz$-crossed ellipsoids. There is only one connected component. For $y$-touching and $z$-touching ellipsoids, the domain hole is a pair of points. The positive and negative parameterizations are sewn together via these two points. In the case of the $yz$-crossed ellipsoids, the domain hole is a unit circle. The two parameterizations are sewn together along a unit circle. In all three sub-cases, the solution is \emph{not} a 2-manifold due to the local topology near domain hole in each case. Excluding the domain hole, the dimension of the solution is $2$. Despite the different geometry of the surfaces in this case, we can use the same parameterization as in the ellipsoidal case 3.

\paragraph{Cases 7 and 8}
The barrel-like parameterization is named according to the shape of the set being parameterized. In this case, one ellipsoid radius is less than $1$, one is less than or equal to $1$, and one is greater than $1$. The barrel-like set under parameterization is part of an ellipsoid, which is located in a unit sphere, as shown in Figure \ref{fig_barrel_param}.
\begin{figure}[ht]
\centering
\begin{subfigure}{0.5\textwidth}
  \centering
  \includegraphics[width=0.95\textwidth]{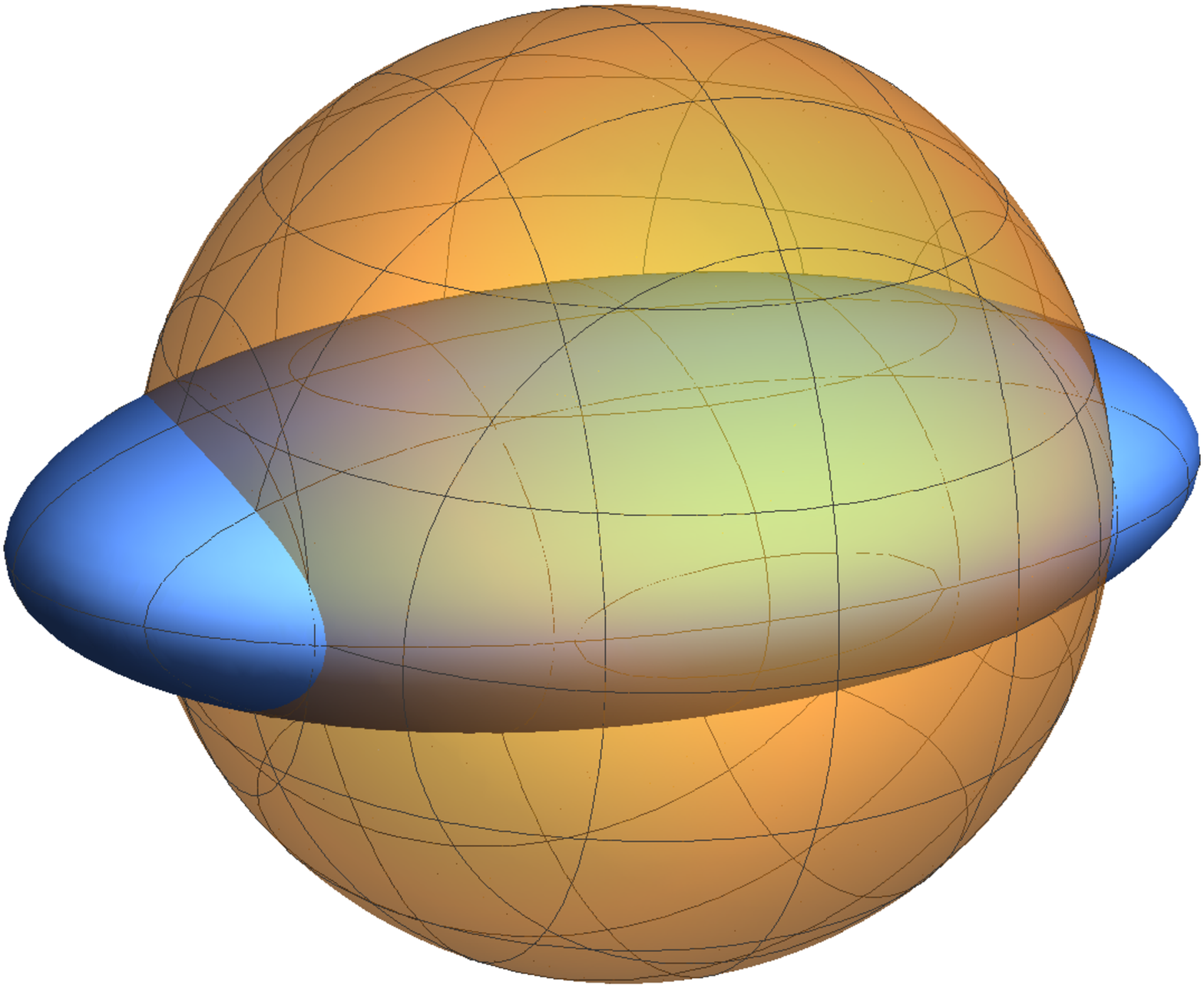}
  \caption{Geometry of parameterization}
\end{subfigure}%
\begin{subfigure}{0.5\textwidth}
  \centering
  \includegraphics[width=0.95\textwidth]{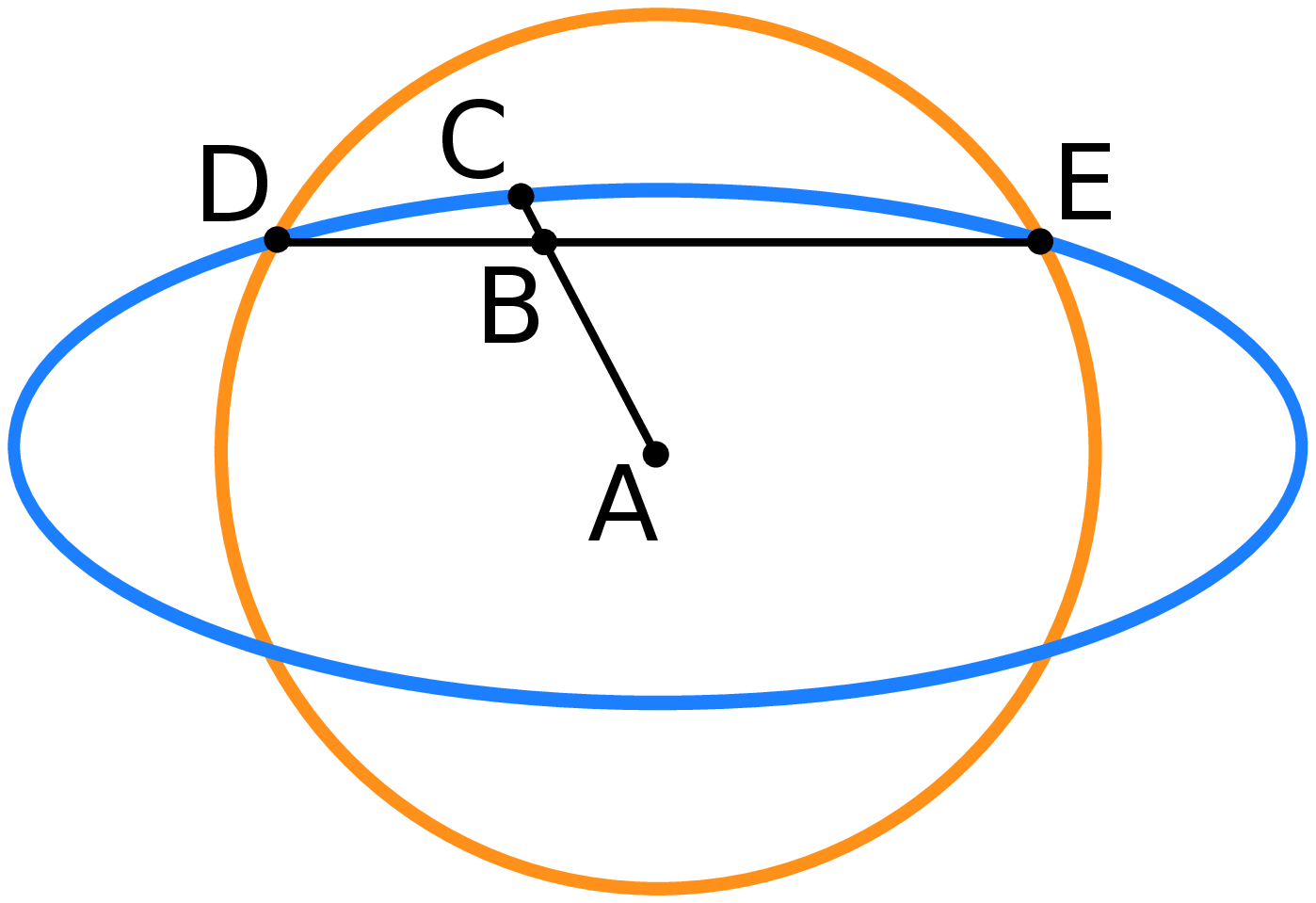}
  \caption{Parameterization schema}
\end{subfigure}%
\caption{Barrel-like parameterization}
\label{fig_barrel_param}
\end{figure}
In order to obtain a parameterization, we calculate two intersection curves (domain holes) on the opposite sides of the longest ellipsoid radius.
These two curves can be parameterized easily with one angle. The corresponding points on both curves $D$ and $E$ are then connected by the segment $DE$.
The segment is parameterized with the second variable. Finally, the vector $AB$ is scaled up by using the implicit equation of the ellipsoid to calculate
a final point on the ellipsoid. For the $y$-barrel-like parameterization, we consider the following intersection curves.
\begin{align} \label{eqn_barrel_cases}
\begin{cases}
& 2 (a + b) \bar{x}^2 + 2 a \bar{y}^2 + 2 b \bar{z}^2 = a + b + c \\
& \bar{x}^2 + \bar{y}^2 + \bar{z}^2 = 1
\end{cases}
\end{align}
A solution to \eqref{eqn_barrel_cases} is a pair of curves:
\begin{align}
\begin{cases}
& \bar{x} = \sqrt{\frac{b - a + c}{2 b}} \cos(\alpha) \\
& \bar{z} = \sqrt{\frac{b - a + c}{2 (b - a)}} \sin(\alpha) \\
& \bar{y} = \pm \sqrt{1 - \bar{x}^2 - \bar{z}^2}
\end{cases}
\end{align}
for $\alpha \in [0, 2 \pi)$. The coordinate $y$ of $DE$ segment is parameterized as follows
\begin{align*}
\tilde{x} & = \bar{x} \\
\tilde{y} & = h \sqrt{1 - \bar{x}^2 - \bar{z}^2} \\
\tilde{z} & = \bar{z}
\end{align*}
for $h \in [-1, 1]$. A factor $d$ of the scaled up $d AB$ segment must satisfy
\begin{equation*}
2 (a + b) (d \tilde{x})^2 + 2 a (d \tilde{y})^2 + 2 b (d \tilde{z})^2 = a + b + c,
\end{equation*}
which reduces to
\begin{equation*}
d = \sqrt{\frac{a + b + c}{2 (a h^2 + (a (1 - h^2) + b) \bar{x}^2 + (b - a h^2) \bar{y}^2)}}
\end{equation*}
and it is sufficient to take only the positive value of $d$ for the purpose of parameterization. After scaling up $AB$ by $d$, the final $y$-barrel parameterization becomes
\begin{align}
\begin{cases}
& x = u k \\
& y = h k \sqrt{1 - u^2 - v^2} \\
& z = v k \\
& w = \pm \sqrt{1 - k^2 (h^2 + (1 - h^2)(u^2 + v^2))},
\end{cases}
\end{align}
where
\begin{align*}
u & = \sqrt{\frac{b - a + c}{2 b}} \cos(\alpha) \\
v & = \sqrt{\frac{b - a + c}{2 (b - a)}} \sin(\alpha) \\
k & = \sqrt{\frac{a + b + c}{2 (a h^2 + (a (1 - h^2) + b) u^2 + (b - a h^2) v^2)}} \\
\alpha & \in [0, 2 \pi) \\
h & \in [-1, 1]
\end{align*}
$z$-barrel case is analogous to $y$-barrel case. The corresponding $z$-barrel parameterization is
\begin{align}
\begin{cases}
& x = u k \\
& y = v k \\
& z = h k \sqrt{1 - u^2 - v^2} \\
& w = \pm \sqrt{1 - k^2 (h^2 + (1 - h^2)(u^2 + v^2))},
\end{cases}
\end{align}
where
\begin{align*}
u & = \sqrt{\frac{a - b + c}{2 a}} \cos(\alpha) \\
v & = \sqrt{\frac{a - b + c}{2 (b - a)}} \sin(\alpha) \\
k & = \sqrt{\frac{a + b + c}{2 (b h^2 + (b (1 - h^2) + a) u^2 + (a - b h^2) v^2)}} \\
\alpha & \in [0, 2 \pi) \\
h & \in [-1, 1].
\end{align*}
There is only one connected component, which comprises a positive and negative parameterization sewn along the two domain holes as deformed ellipses. Roughly speaking, a moving point on a positive parameterization that goes through the first domain hole enters the negative parameterization. The point then enters the opposite domain hole on the negative parameterization and returns to the positive parameterization through the second domain hole. From a topological viewpoint, the domain is a torus and the dimension of the solution is $2$.

\paragraph{Cases 9 and 10}

The two cases of notched barrel parameterization are similar to the barrel parameterization case. The only difference is that a pair of antipodal points exists on the ellipsoid where the domain is shrunk to a point. As a result, the spin-quadric is not a $2$-manifold near these two points. The domain hole is also unusual because it comprises two deformed ellipses with two common points. A notched barrel case is shown in Figure \ref{fig_notched_barrel_param}.

\begin{figure}[ht]
\centering
\includegraphics[width=0.5\textwidth]{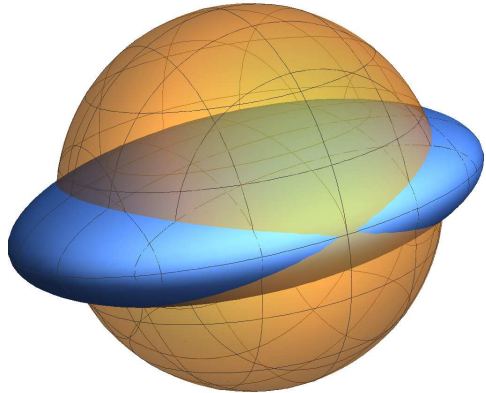}
\caption{Notched-barrel-like parameterization}
\label{fig_notched_barrel_param}
\end{figure}

Despite the different geometric properties of the notched-barrel-like parameterization, we can use the same parameterization employed for the barrel-like case of parameterization.

\paragraph{Case 11}
The last case is a $yz$-caps parameterization, which geometrically resembles a pair of caps on the opposite sides of an ellipsoid.
In this case, one radius is less than or equal to $1$ and the remaining two radii are greater than one. Both caps need to be parameterized separately because they never have a common point.
A schema for $yz$-caps parameterization is presented in Figure \ref{fig_caps_param}.
\begin{figure}[ht]
\centering
\begin{subfigure}{0.5\textwidth}
  \centering
  \includegraphics[width=0.95\textwidth]{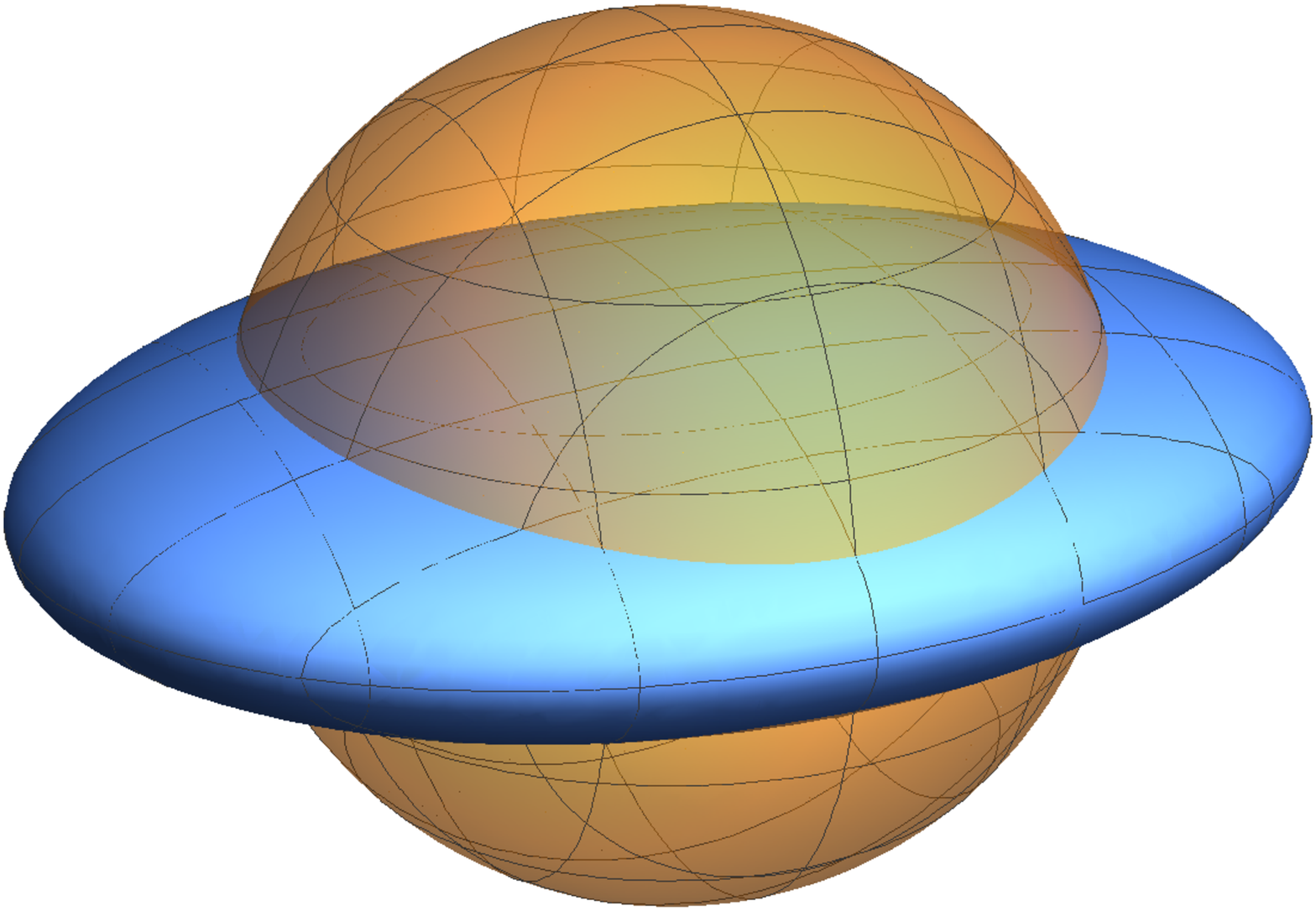}
  \caption{Geometry of parameterization}
\end{subfigure}%
\begin{subfigure}{0.5\textwidth}
  \centering
  \includegraphics[width=0.95\textwidth]{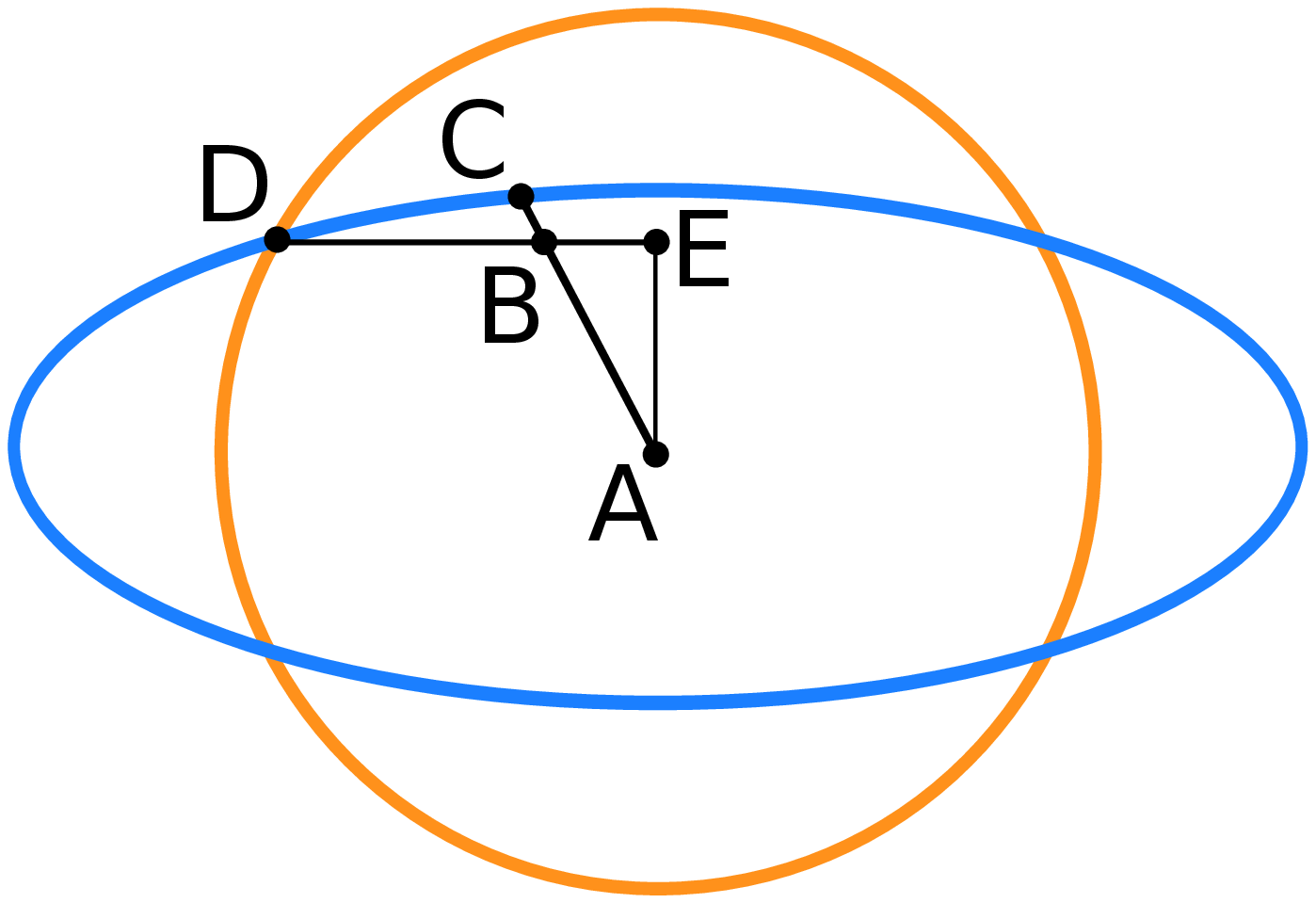}
  \caption{Parameterization schema}
\end{subfigure}%
\caption{Caps-like parameterization}
\label{fig_caps_param}
\end{figure}
A parameterization is derived in a similar manner to the barrel-like case. First, two domain hole curves are calculated and parameterized with an angle.
Next, the segment $DE$ is introduced, which connects each point on a domain hole curve with the point $E$. This is the central point of the domain hole projection onto the $yz$ plane.
The point $B$ on segment $DE$ is parameterized with the second variable. To generate a parametric point on the ellipsoid, we connect
segment $AB$ and point $B$ with the central point $A$ of the ellipsoid. The final point $C$ on a line extending $AB$ is then calculated from the ellipsoid equation.
For $yz$-caps parameterization, we consider the following intersection curves.
\begin{align} \label{eqn_caps_cases}
\begin{cases}
& 2 (a + b) \bar{x}^2 + 2 a \bar{y}^2 + 2 b \bar{z}^2 = a + b + c \\
& \bar{x}^2 + \bar{y}^2 + \bar{z}^2 = 1
\end{cases}
\end{align}
A solution to \eqref{eqn_caps_cases} is a pair of curves:
\begin{align}
\begin{cases}
& \bar{y} = \sqrt{\frac{a + b - c}{2 b}} \cos(\alpha) \\
& \bar{z} = \sqrt{\frac{a + b - c}{2 a}} \sin(\alpha) \\
& \bar{x} = \sigma \sqrt{1 - \bar{y}^2 - \bar{z}^2}
\end{cases}
\end{align}
for $\alpha \in [0, 2 \pi)$ and $\sigma \in \{-1, 1\}$. For each domain hole $\sigma$, the segment $DE$ is parameterized with the second variable:
\begin{align}
\begin{cases}
& \tilde{x} = \bar{x} \\
& \tilde{y} = h \bar{y} \\
& \tilde{z} = h \bar{z}
\end{cases}
\end{align}
for $h \in [0, 1]$. Note that the coordinates of the segment $AB$ are $[\tilde{x}, \tilde{y}, \tilde{z}]^T$. In order to calculate $C$, we find $d$ such that $C = d AB$ is a point on the ellipsoid:
\begin{equation*}
2 (a + b) (d \tilde{x})^2 + 2 a (d \tilde{y})^2 + 2 b (d \tilde{z})^2 = a + b + c,
\end{equation*}
which reduces to
\begin{equation*}
d = \sqrt{\frac{a + b + c}{2 (a + b - (a (1 - h^2) + b) \bar{y}^2 - (b (1 - h^2) + a) \bar{z}^2}}
\end{equation*}
and only the positive value of $d$ is taken for parameterization. After scaling up $DE$ by $d$, the final yz-caps parameterization becomes
\begin{align}
\begin{cases}
& x = \sigma k \sqrt{1 - u^2 - v^2} \\
& y = h k u \\
& z = h k v \\
& w = \pm \sqrt{1 - k^2 (1 - (1 - h^2)(u^2 + v^2))},
\end{cases}
\end{align}
where
\begin{align*}
u & = \sqrt{\frac{a + b - c}{2 b}} \cos(\alpha) \\
v & = \sqrt{\frac{a + b - c}{2 a}} \sin(\alpha) \\
k & = \sqrt{\frac{a + b + c}{2 (a + b - (a (1 - h^2) + b) u^2 - (b (1 - h^2) + a) v^2}} \\
\alpha & \in [0, 2 \pi) \\
h & \in [0, 1] \\
\sigma & \in \{-1, 1\},
\end{align*}
There are two connected components, where each comprises a positive and negative parameterization sewn along a corresponding domain hole as a deformed ellipse. The dimension of the solution is $2$.

All of the cases have been considered so the parameterization and the proof are complete.
\end{proof}

An ellipsoidal general predicate is named according to its parameterization domain, which is an ellipsoid or part of it. In this section, we conclude with the parameterization of an ellipsoidal general predicate.
\begin{theorem} \label{thm_ellipsoidal_predicate_parameterization}
Let $\GG_s = P \cdot \Rot_s(Q) + U \cdot \Rot_s(V) + c$ be \textbf{an ellipsoidal general predicate}. Assume that $S_\GG(s)$ is a spin-quadric representing $\GG_s$ and $M_\GG$ is the matrix of the associated quadratic form.
Let $\mathbf{s}(\omega) = [s_{12}(\omega), s_{23}(\omega), s_{31}(\omega), s_0(\omega)]^T$ be a vectorized form of a spinor $s(\omega) = s_{12}(\omega) \e_{12} + s_{23}(\omega) \e_{23}(\omega) + s_{31}(\omega) \e_{31} + s_0(\omega)$ and $\mathbf{t}$ is a vector $\mathbf{t}(\omega) = [x(\omega), y(\omega), z(\omega), w(\omega)]^T$ for $\omega \in \Omega \subset \R^2$. A parameterization of the spin-quadric $S_\GG = s(\omega), \, \omega \in \Omega$ exists:
\begin{equation}
\mathbf{s} = Q \mathbf{t},
\end{equation}
which comprises $11$ parameterization cases for $\mathbf{t}$ according to Lemma \ref{lem_solutions_ellipsoidal}, where the parameters $a$ and $b$ are set to
\begin{align*}
a & = \|P\|\|Q\| \\
b & = \|U\|\|V\|
\end{align*}
and a rotation matrix $Q$, which is a matrix of orthonormal eigenvectors of matrix $M_\GG$:
\begin{equation*}
Q =
\left[ \begin{smallmatrix}
    \frac{W^{(1)}_1}{\|W^{(1)}\|} & \frac{W^{(2)}_1}{\|W^{(2)}\|} & \frac{W^{(3)}_1}{\|W^{(3)}\|} & \frac{W^{(4)}_1}{\|W^{(4)}\|} \\
    \frac{W^{(1)}_2}{\|W^{(1)}\|} & \frac{W^{(2)}_2}{\|W^{(2)}\|} & \frac{W^{(3)}_2}{\|W^{(3)}\|} & \frac{W^{(4)}_2}{\|W^{(4)}\|} \\
    \frac{W^{(1)}_3}{\|W^{(1)}\|} & \frac{W^{(2)}_3}{\|W^{(2)}\|} & \frac{W^{(3)}_3}{\|W^{(3)}\|} & \frac{W^{(4)}_3}{\|W^{(4)}\|} \\
    \frac{W^{(1)}_4}{\|W^{(1)}\|} & \frac{W^{(2)}_4}{\|W^{(2)}\|} & \frac{W^{(3)}_4}{\|W^{(3)}\|} & \frac{W^{(4)}_4}{\|W^{(4)}\|}
\end{smallmatrix} \right],
\end{equation*}
where the vectors $W^{(j)}$ are calculated according to Lemma \ref{lem_eigenvectors_of_general_predicate}.
\end{theorem}

The theorem above gives an insight into the global geometry of an ellipsoidal general predicate. The following are corollaries of the theorem.
\begin{itemize}
\item A spin-quadric of an ellipsoidal general predicate can be empty, zero-dimensional, or two-dimensional.
\item A spin-quadric of an ellipsoidal general predicate can comprise zero, one, or two connected components.
\item The genus of a spin-quadric of an ellipsoidal general predicate can be zero or one.
\item Ellipsoidal general predicates exist such that the associated spin-quadric is not a manifold.
\end{itemize}

\subsubsection{Parameterization of a toroidal general predicate}

For a toroidal general predicate, either $a = 0$ or $b = 0$, but not $a = b = 0$. In this case, equations \eqref{eqn_parameterization_cases_substituted} are reduced to
\begin{align} \label{eqn_parameterization_cases_toroidal_b_zero}
\begin{cases}
& a x^2 + a y^2 - a z^2 - a w^2 = c \\
& x^2 + y^2 + z^2 + w^2 = 1
\end{cases}
\end{align}
or
\begin{align} \label{eqn_parameterization_cases_toroidal_a_zero}
\begin{cases}
& b x^2 - b y^2 + b z^2 - b w^2 = c \\
& x^2 + y^2 + z^2 + w^2 = 1
\end{cases}
\end{align}
for $b = 0$ and $a = 0$, respectively. Lemma \ref{lem_solutions_toroidal} gives a solution to the set of equations \eqref{eqn_parameterization_cases_toroidal_b_zero} and \eqref{eqn_parameterization_cases_toroidal_a_zero}.
\begin{lemma} \label{lem_solutions_toroidal}
A solution to the following set of equations involving $x, y, z, w$
\begin{align*}
\begin{cases}
& (a + b) x^2 + (a - b) y^2 + (-a + b) z^2 + (-a - b) w^2 = c \\
& x^2 + y^2 + z^2 + w^2 = 1
\end{cases}
\end{align*}
for parameters $a \in \R^+, b \in \R^+, c \in \R$ such that
\begin{equation}
(a = 0 \wedge b \ne 0) \vee (a \ne 0 \wedge b = 0)
\end{equation}
can be parameterized according to the formulae given in the proof.
\end{lemma}
\begin{proof}
There are two major cases: the first for $a \ne 0$ and the second for $b \ne 0$, where these cases give the following set of equations, respectively.
\begin{align*}
\begin{cases}
& a x^2 + a y^2 - a z^2 - a w^2 = c \\
& x^2 + y^2 + z^2 + w^2 = 1
\end{cases}
\end{align*}
and
\begin{align*}
\begin{cases}
& b x^2 - b y^2 + b z^2 - b w^2 = c \\
& x^2 + y^2 + z^2 + w^2 = 1
\end{cases}
\end{align*}
\paragraph{First case} Consider the case where $a \ne 0$. Since $a$ is non-zero, both sides of the first equation can be divided by $a$.
By replacing the two equations with their sum and their difference, the set of equations is separated based on the variables, as follows.
\begin{align}
\begin{cases} \label{tor_cases_first}
& x^2 + y^2 = \frac{a + c}{2 a} \\
& z^2 + w^2 = \frac{a - c}{2 a}
\end{cases}
\end{align}
A solution exists unless
\begin{equation}
\frac{a + c}{2 a} < 0 \, \vee \, \frac{a - c}{2 a} < 0 \iff c \in \mathopen{}\mathclose{\left(-\infty, -a \right)} \cup \mathopen{}\mathclose{\left(a, \infty \right)},
\end{equation}
which defines a sub-case as follows.
\begin{itemize}
    \item Sub-case labelled \emph{"an empty set"} \\
          equivalent condition: $c \in \mathopen{}\mathclose{\left(-\infty, -a \right)} \cup \mathopen{}\mathclose{\left(a, \infty \right)}$
\end{itemize}
If $c = -a$ or $c = a$, the solution to \ref{tor_cases_first} is reduced to a circle, which corresponds to the following two sub-cases:
\begin{itemize}
    \item Sub-case labelled \emph{"a $zw$-circle"} \\
          equivalent condition: $c = -a$
    \item Sub-case labelled \emph{"a $xy$-circle"} \\
          equivalent condition: $c = a$
\end{itemize}
otherwise, the solution to \ref{tor_cases_first} has a toroidal parameterization, as follows.
\begin{itemize}
    \item Sub-case labelled \emph{"a $xy/zw$-torus"} \\
          equivalent condition: $c \in \mathopen{}\mathclose{\left(-a, a \right)}$
\end{itemize}
\paragraph{Second case} Consider the case where $b \ne 0$. Since $b$ is non-zero, both sides of the first equation can be divided by $b$.
By replacing the two equations with their sum and difference, the set of equations is separated based on the variables,
\begin{align}
\begin{cases} \label{tor_cases_second}
& x^2 + z^2 = \frac{b + c}{2 b} \\
& y^2 + w^2 = \frac{b - c}{2 b}
\end{cases}
\end{align}
A solution exists unless
\begin{equation}
\frac{b + c}{2 b} < 0 \, \vee \, \frac{b - c}{2 b} < 0 \iff c \in \mathopen{}\mathclose{\left(-\infty, -b \right)} \cup \mathopen{}\mathclose{\left(b, \infty \right)},
\end{equation}
which defines a sub-case, as follows.
\begin{itemize}
    \item Sub-case labelled \emph{"an empty set"} \\
          equivalent condition: $c \in \mathopen{}\mathclose{\left(-\infty, -b \right)} \cup \mathopen{}\mathclose{\left(b, \infty \right)}$
\end{itemize}
If $c = -b$ or $c = b$, the solution to \ref{tor_cases_second} is reduced to a circle, which corresponds to the following two sub-cases:
\begin{itemize}
    \item Sub-case labelled \emph{"a $yw$-circle"} \\
          equivalent condition: $c = -b$
    \item Sub-case labelled \emph{"a $xz$-circle"} \\
          equivalent condition: $c = b$
\end{itemize}
otherwise, the solution to \ref{tor_cases_second} has a toroidal parameterization, as follows.
\begin{itemize}
    \item Sub-case labelled \emph{"a $xz/yw$-torus"} \\
          equivalent condition: $c \in \mathopen{}\mathclose{\left(-b, b \right)}$
\end{itemize}
Depending on the values of $a$ and $b$, there are $7$ different cases of parameterization, which are summarized in Table \ref{tor_param_domain}.
Note that in contrast to Lemma \ref{lem_solutions_ellipsoidal}, there are no domain holes.
\begin{table}[ht]
\begin{center}
\begin{tabular}{|c|l|l|l|}
\hline
param. type & $a \ne 0$ & $b \ne 0$ \\ \hline
\textbf{1/1} & $c \in \mathopen{}\mathclose{\left(-\infty, -a \right)}$ & $c \in \mathopen{}\mathclose{\left(-\infty, -b \right)}$ \\ \hline
\textbf{4/7} & $c = - a$ & $c = - b$ \\ \hline
\textbf{2/5} & $c \in \mathopen{}\mathclose{\left(-a, a \right)}$ & $c \in \mathopen{}\mathclose{\left(- b, b \right)}$ \\ \hline
\textbf{3/6} & $c = a$ & $c = b$ \\ \hline
\textbf{1/1} & $c \in \mathopen{}\mathclose{\left(a, \infty \right)}$ & $c \in \mathopen{}\mathclose{\left(b, \infty \right)}$ \\ \hline
\end{tabular}
\end{center}
\caption{Parameterization types for different values of $a$ and $b$}
\label{tor_param_domain}
\end{table}
The cases are numbered according to the following list of \emph{parameterization types}.
\begin{enumerate}
\item[(1)] an empty case
\item[(2)] a $xy/zw$-torus
\item[(3)] a $xy$-circle
\item[(4)] a $zw$-circle
\item[(5)] a $xz/yw$-torus
\item[(6)] a $xz$-circle
\item[(7)] a $yw$-circle
\end{enumerate}

\paragraph{Case 1}
There is no solution or parameterization.

\paragraph{Cases 3, 4, 6, and 7}
A solution to \ref{tor_cases_first} and \ref{tor_cases_second} is an axial aligned circle. The corresponding parameterizations are \\
for case 3:
\begin{align*}
\begin{cases}
& x = \sqrt{\frac{a + c}{2 a}} \cos(\alpha) \\
& y = \sqrt{\frac{a + c}{2 a}} \sin(\alpha) \\
& z = 0 \\
& w = 0
\end{cases}
\end{align*}
for case 4:
\begin{align*}
\begin{cases}
& x = 0 \\
& y = 0 \\
& z = \sqrt{\frac{a - c}{2 a}} \cos(\alpha) \\
& w = \sqrt{\frac{a - c}{2 a}} \sin(\alpha)
\end{cases}
\end{align*}
for case 6:
\begin{align*}
\begin{cases}
& x = \sqrt{\frac{b + c}{2 b}} \cos(\alpha) \\
& y = 0 \\
& z = \sqrt{\frac{b + c}{2 b}} \sin(\alpha) \\
& w = 0
\end{cases}
\end{align*}
for case 7:
\begin{align*}
\begin{cases}
& x = 0 \\
& y = \sqrt{\frac{b - c}{2 b}} \cos(\alpha) \\
& z = 0 \\
& w = \sqrt{\frac{b - c}{2 b}} \sin(\alpha),
\end{cases}
\end{align*}
where $\alpha \in [0, 2 \pi)$ in all cases.

\paragraph{Cases 2 and 5}
The set under parameterization is a torus. A common trigonometric parameterization is \\
for case 2:
\begin{align*}
\begin{cases}
& x = \sqrt{\frac{a + c}{2 a}} \cos(\alpha) \\
& y = \sqrt{\frac{a + c}{2 a}} \sin(\alpha) \\
& z = \sqrt{\frac{a - c}{2 a}} \cos(\beta) \\
& w = \sqrt{\frac{a - c}{2 a}} \sin(\beta)
\end{cases}
\end{align*}
for case 5:
\begin{align*}
\begin{cases}
& x = \sqrt{\frac{b + c}{2 b}} \cos(\alpha) \\
& y = \sqrt{\frac{b - c}{2 b}} \cos(\beta) \\
& z = \sqrt{\frac{b + c}{2 b}} \sin(\alpha) \\
& w = \sqrt{\frac{b - c}{2 b}} \sin(\beta),
\end{cases}
\end{align*}
where $\alpha \in [0, 2 \pi)$, $\beta \in [0, 2 \pi)$ in all cases.
\end{proof}

A toroidal general predicate is named according to its parameterization domain, which is a torus or a circle. This section concludes with a parameterization of a toroidal general predicate.
\begin{theorem} \label{thm_toroidal_predicate_parameterization}
Let $\GG_s = P \cdot \Rot_s(Q) + U \cdot \Rot_s(V) + c$ be \textbf{a toroidal general predicate}. Assume that $S_\GG(s)$ is a spin-quadric representing $\GG_s$ and $M_\GG$ is the matrix of the associated quadratic form.
Let $\mathbf{s}(\omega) = [s_{12}(\omega), s_{23}(\omega), s_{31}(\omega), s_0(\omega)]^T$ be a vectorized form of a spinor $s(\omega) = s_{12}(\omega) \e_{12} + s_{23}(\omega) \e_{23}(\omega) + s_{31}(\omega) \e_{31} + s_0(\omega)$ and $\mathbf{t}$ is a vector $\mathbf{t}(\omega) = [x(\omega), y(\omega), z(\omega), w(\omega)]^T$ for $\omega \in \Omega \subset \R^2$. A parameterization of the spin-quadric $S_\GG = s(\omega), \, \omega \in \Omega$ exists:
\begin{equation}
\mathbf{s} = Q \mathbf{t},
\end{equation}
which comprises $7$ parameterization cases for $\mathbf{t}$ according to Lemma \ref{lem_solutions_toroidal}, where the parameters $a$ and $b$ are set to
\begin{align*}
a & = \|P\|\|Q\| \\
b & = \|U\|\|V\|
\end{align*}
and a rotation matrix $Q$, which is a matrix of orthonormal eigenvectors of matrix $M_\GG$:
\begin{equation*}
Q =
\left[ \begin{smallmatrix}
    \frac{W^{(1)}_1}{\|W^{(1)}\|} & \frac{W^{(2)}_1}{\|W^{(2)}\|} & \frac{W^{(3)}_1}{\|W^{(3)}\|} & \frac{W^{(4)}_1}{\|W^{(4)}\|} \\
    \frac{W^{(1)}_2}{\|W^{(1)}\|} & \frac{W^{(2)}_2}{\|W^{(2)}\|} & \frac{W^{(3)}_2}{\|W^{(3)}\|} & \frac{W^{(4)}_2}{\|W^{(4)}\|} \\
    \frac{W^{(1)}_3}{\|W^{(1)}\|} & \frac{W^{(2)}_3}{\|W^{(2)}\|} & \frac{W^{(3)}_3}{\|W^{(3)}\|} & \frac{W^{(4)}_3}{\|W^{(4)}\|} \\
    \frac{W^{(1)}_4}{\|W^{(1)}\|} & \frac{W^{(2)}_4}{\|W^{(2)}\|} & \frac{W^{(3)}_4}{\|W^{(3)}\|} & \frac{W^{(4)}_4}{\|W^{(4)}\|}
\end{smallmatrix} \right],
\end{equation*}
where the vectors $W^{(j)}$ are calculated according to Lemma \ref{lem_eigenvectors_of_general_predicate}.
\end{theorem}

The theorem above gives an insight into the global geometry of a toroidal general predicate. The following are the corollaries of the theorem.
\begin{itemize}
\item A spin-quadric of a toroidal general predicate can be empty, one-dimensional, or two-dimensional.
\item A spin-quadric of a toroidal general predicate always comprises one connected component.
\item The genus of a spin-quadric of a toroidal general predicate is always one.
\item A spin-quadric associated with a toroidal general predicate is always a manifold.
\end{itemize}

\section{Results and discussion}

Experimental tests were conducted by using the \emph{libcs2} \cite{libcs2} library to implement a complete parameterization of a general predicate.
The pictures were generated using the \emph{vis} visualization application, which is part of the \emph{libcs2} library.
In \emph{vis} stereographic projection, $\R^4 \supset \S^3 \longrightarrow \R^3$ was used because this projection does not generate artificial three-dimensional intersections.
Changing the parameter $c$ from $-\infty$ to $\infty$ yielded a sequence of transitions between different types of spin-quadric parameterizations.
In this sequence, the most probable cases were an ellipsoidal parameterization, barrel-like parameterization, and caps-like parameterization, where these three cases are presented in Figures \ref{fig_typical_ellipsoids_param}, \ref{fig_typical_barrel_param}, and \ref{fig_typical_caps_param}, respectively. An ellipsoidal parameterization gives a pair of separate deformed ellipsoids \ref{fig_typical_ellipsoids_param}.
\begin{figure}[ht]
\centering
\begin{subfigure}{0.5\textwidth}
  \centering
  \includegraphics[width=0.95\textwidth]{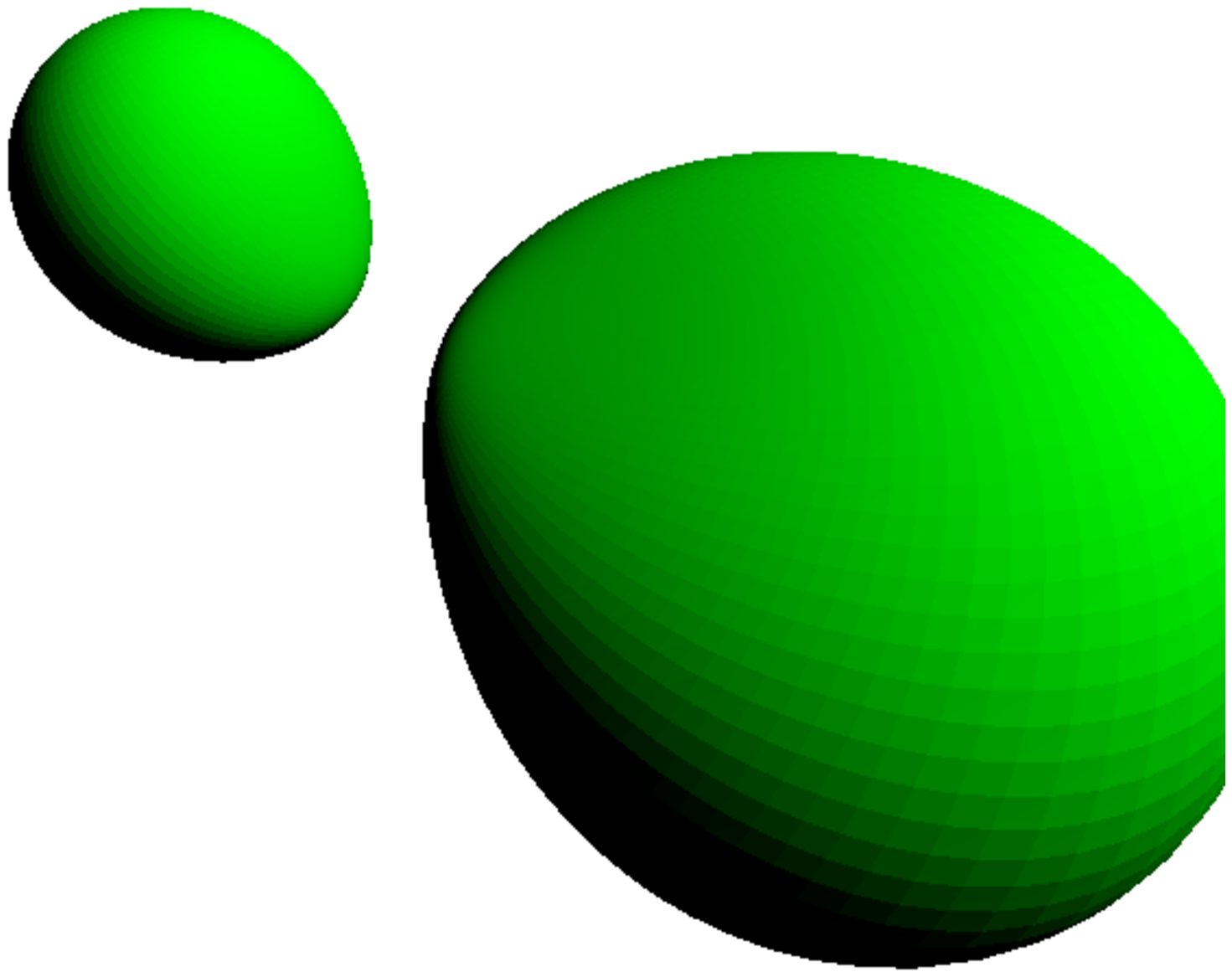}
\end{subfigure}%
\begin{subfigure}{0.5\textwidth}
  \centering
  \includegraphics[width=0.95\textwidth]{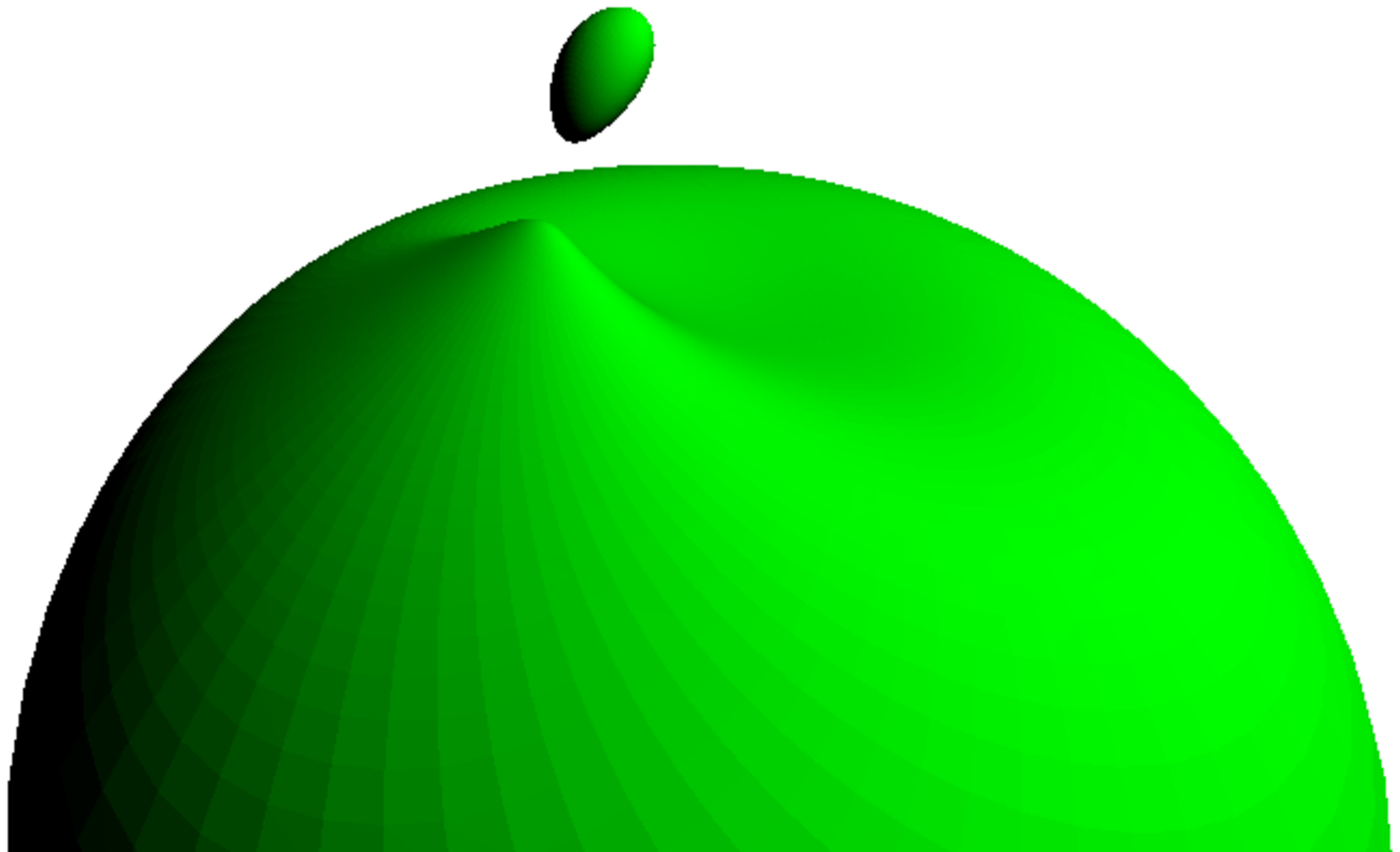}
\end{subfigure}%
\caption{Two examples of spin-quadrics with an ellipsoidal parameterization}
\label{fig_typical_ellipsoids_param}
\end{figure}
From a topological viewpoint, the most interesting case is a spin-quadric with a barrel-like parameterization \ref{fig_typical_barrel_param}.
\begin{figure}[ht]
\centering
\begin{subfigure}{0.5\textwidth}
  \centering
  \includegraphics[width=0.95\textwidth]{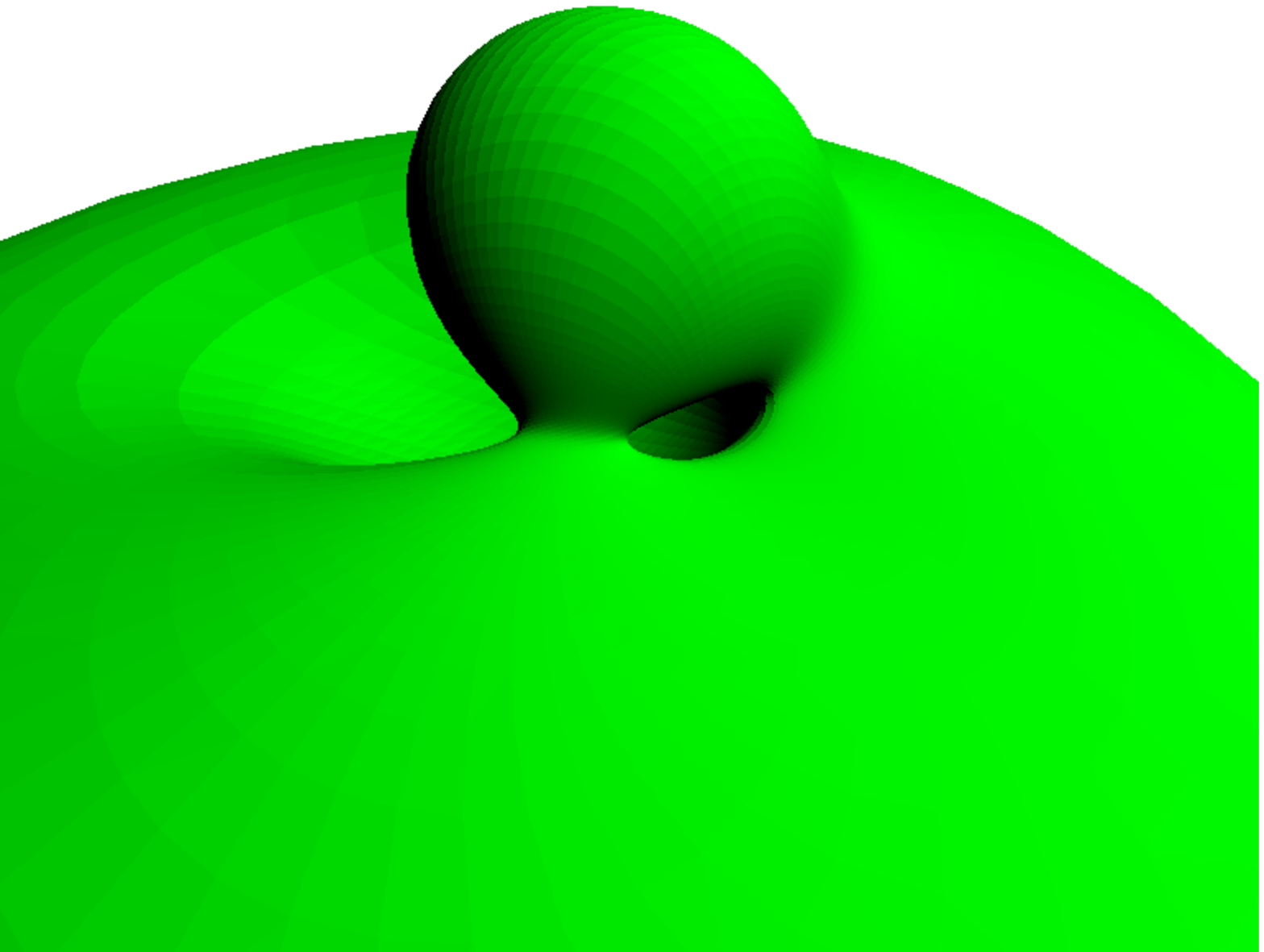}
\end{subfigure}%
\begin{subfigure}{0.5\textwidth}
  \centering
  \includegraphics[width=0.95\textwidth]{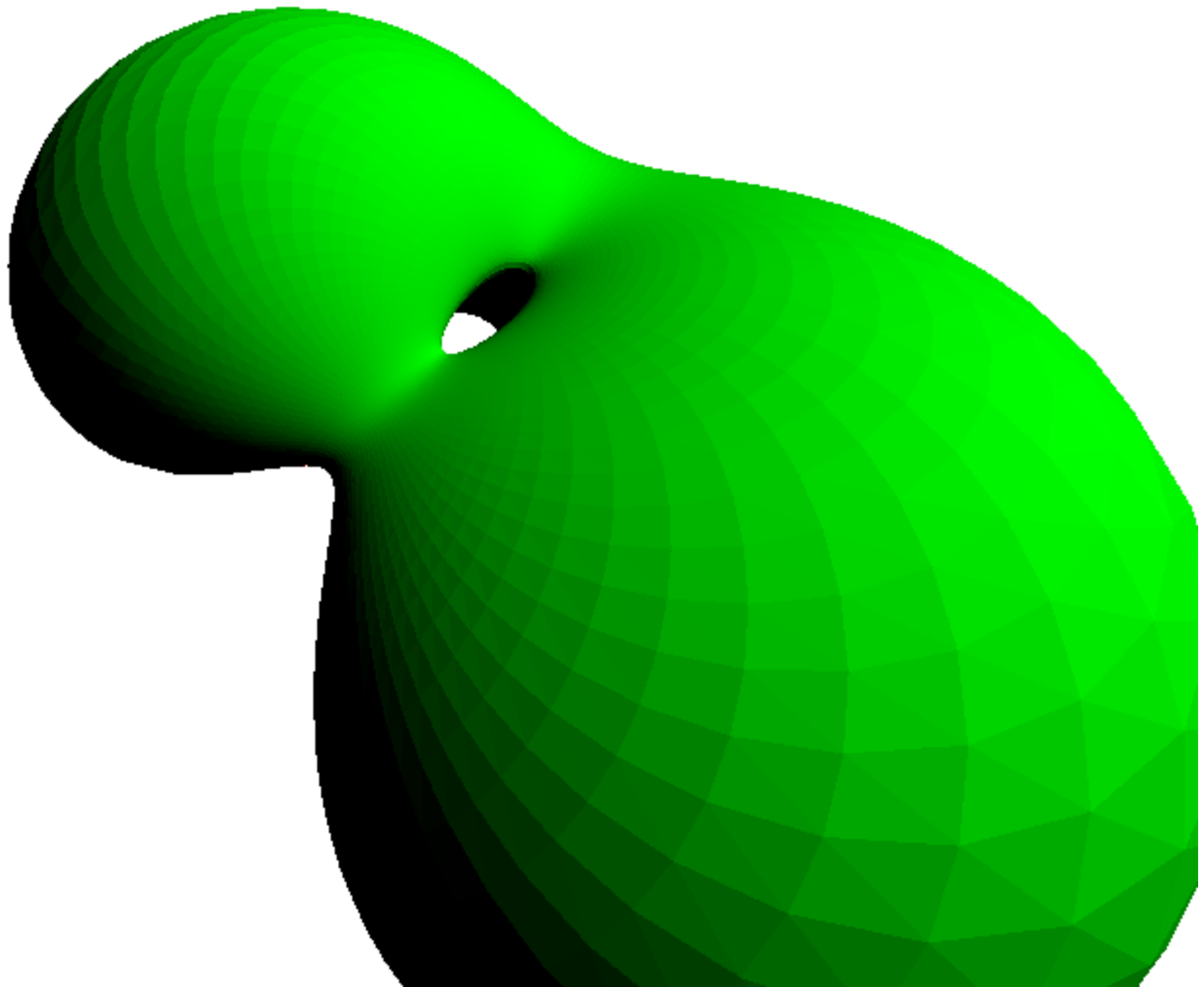}
\end{subfigure}%
\caption{Two examples of spin-quadrics with a barrel-like parameterization}
\label{fig_typical_barrel_param}
\end{figure}
The spin-quadrics generated by caps-like parameterization are much more deformed than those generated by ellipsoidal parameterization \ref{fig_typical_caps_param}.
\begin{figure}[ht]
\centering
\begin{subfigure}{0.5\textwidth}
  \centering
  \includegraphics[width=0.95\textwidth]{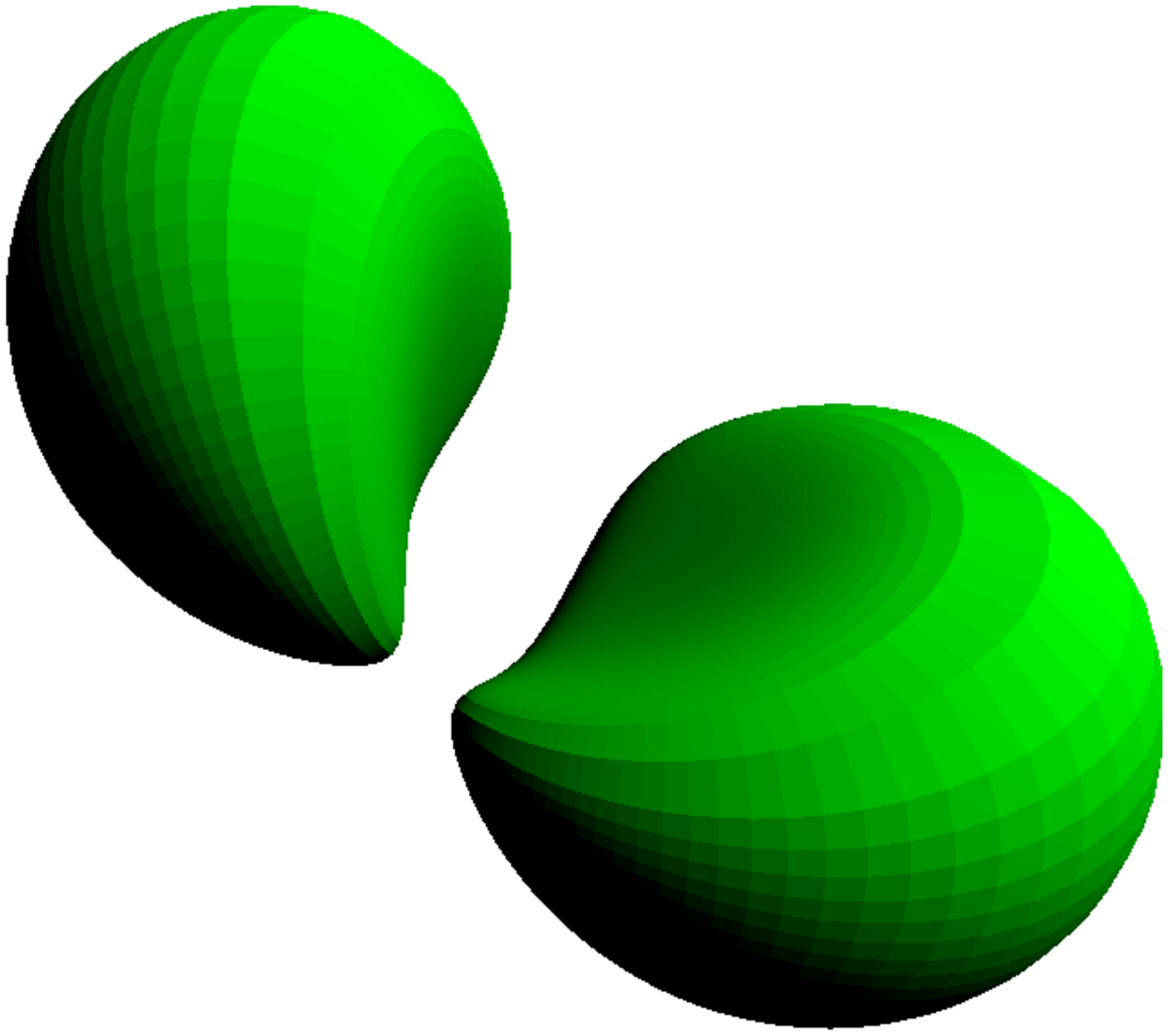}
\end{subfigure}%
\begin{subfigure}{0.5\textwidth}
  \centering
  \includegraphics[width=0.95\textwidth]{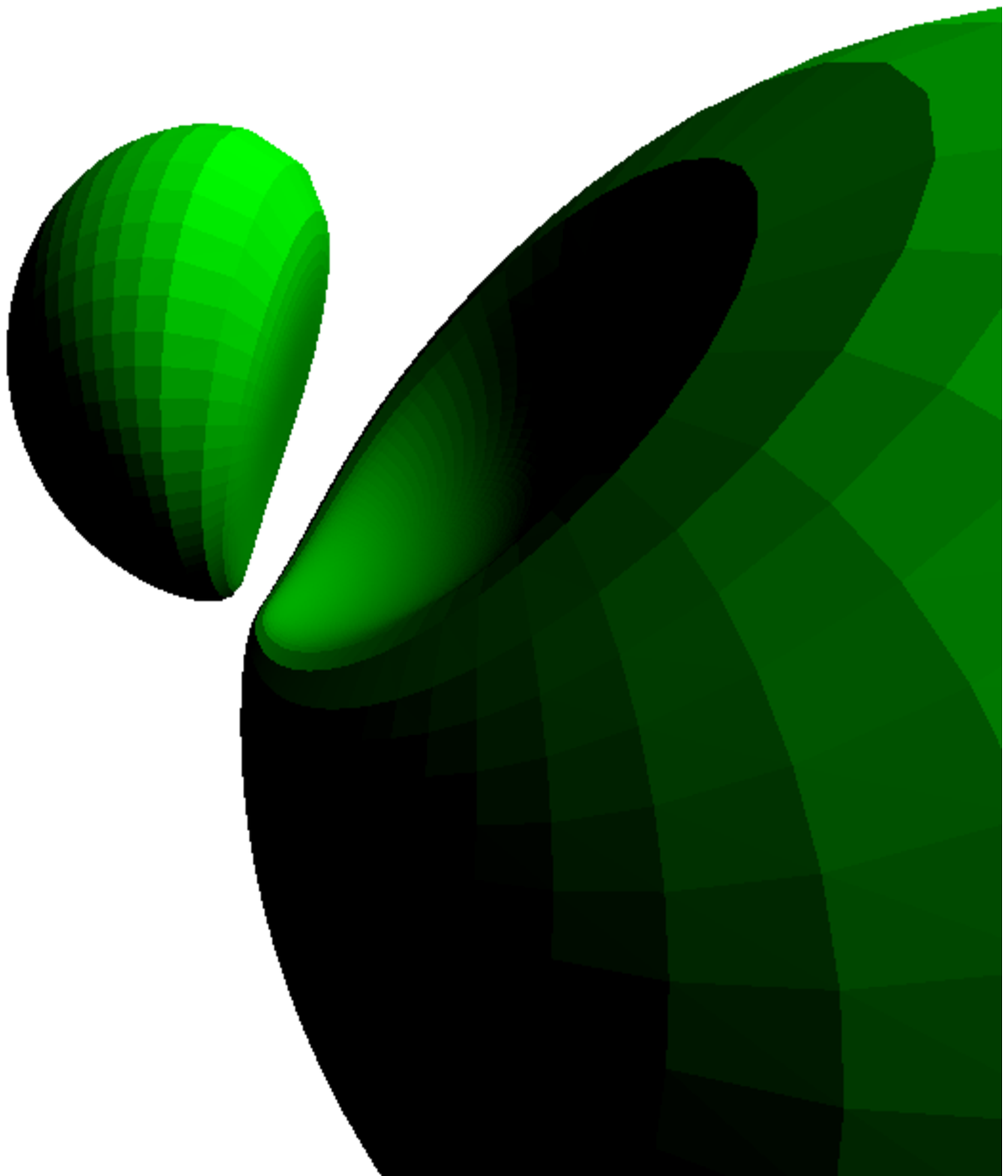}
\end{subfigure}%
\caption{Two examples of spin-quadrics with a caps-like parameterization}
\label{fig_typical_caps_param}
\end{figure}
Toroidal parameterization is straightforward and a spin-quadric with a toroidal parameterization is actually a \emph{Clifford torus} in $\S^4$ (specifically, a generalization of \emph{a Clifford torus}). After projection, it becomes a standard torus in three dimensions \ref{fig_typical_torus_param}.
\begin{figure}[ht]
\centering
\begin{subfigure}{0.5\textwidth}
  \centering
  \includegraphics[width=0.95\textwidth]{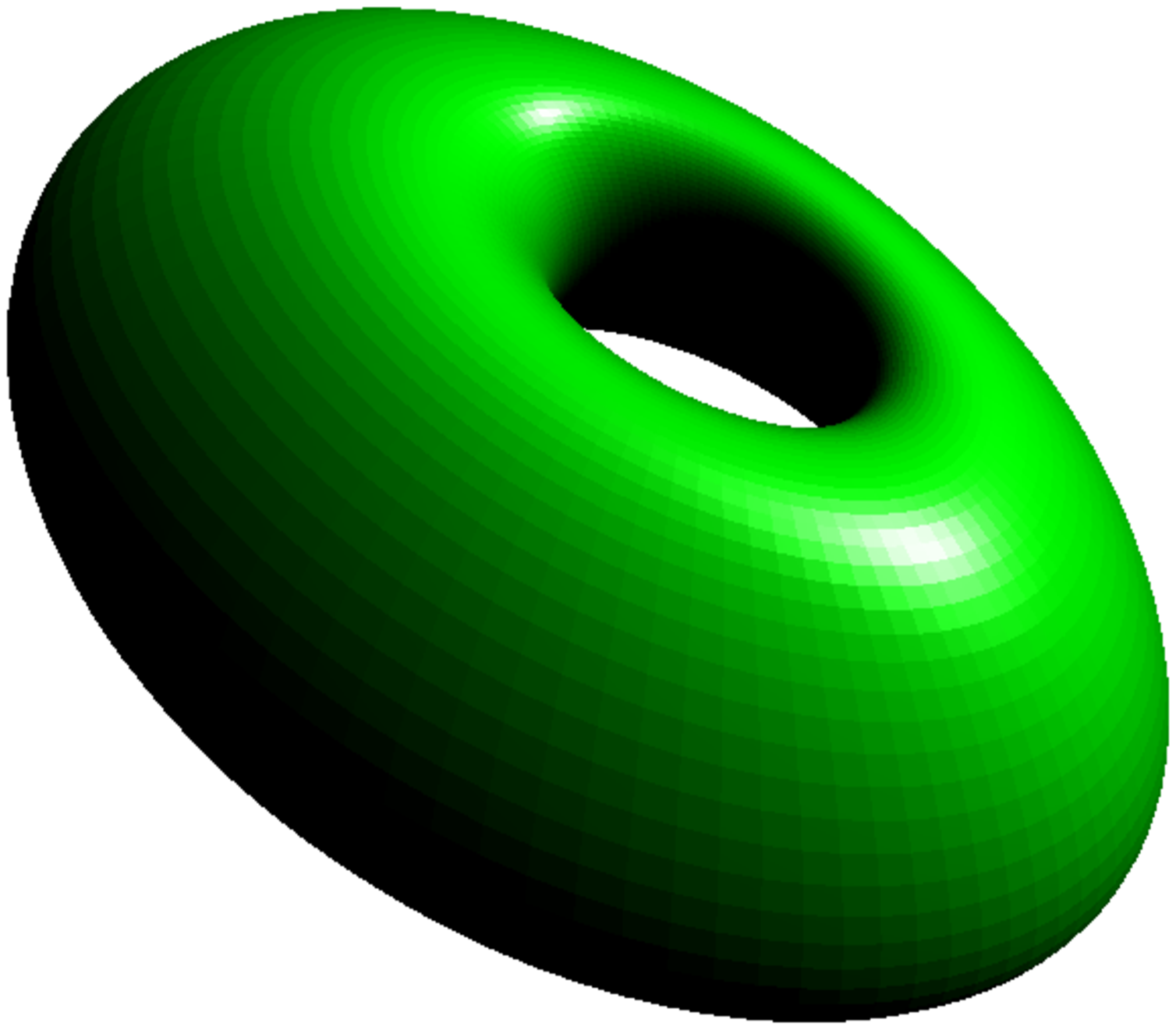}
\end{subfigure}%
\begin{subfigure}{0.5\textwidth}
  \centering
  \includegraphics[width=0.95\textwidth]{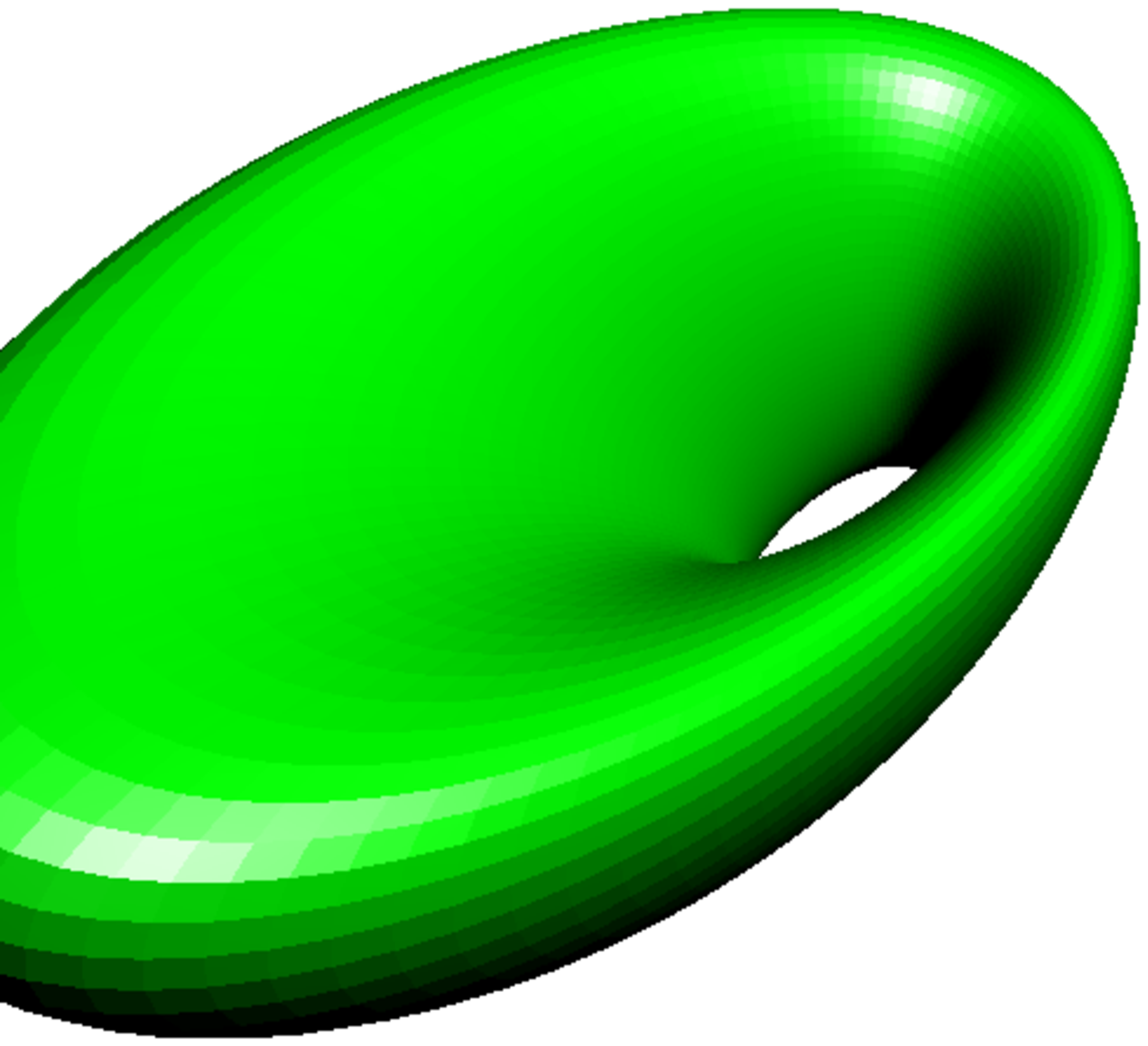}
\end{subfigure}%
\caption{Two examples of spin-quadrics with a torus-like parameterization}
\label{fig_typical_torus_param}
\end{figure}
There are also some unusual cases of parameterization such as $yz$-crossed ellipsoids and an ellipsoidal parameterization, where one of the ellipsoids is inside the other. These two parameterization examples are presented in Figure \ref{fig_not_typical_param}.
\begin{figure}[ht]
\centering
\begin{subfigure}{0.5\textwidth}
  \centering
  \includegraphics[width=0.95\textwidth]{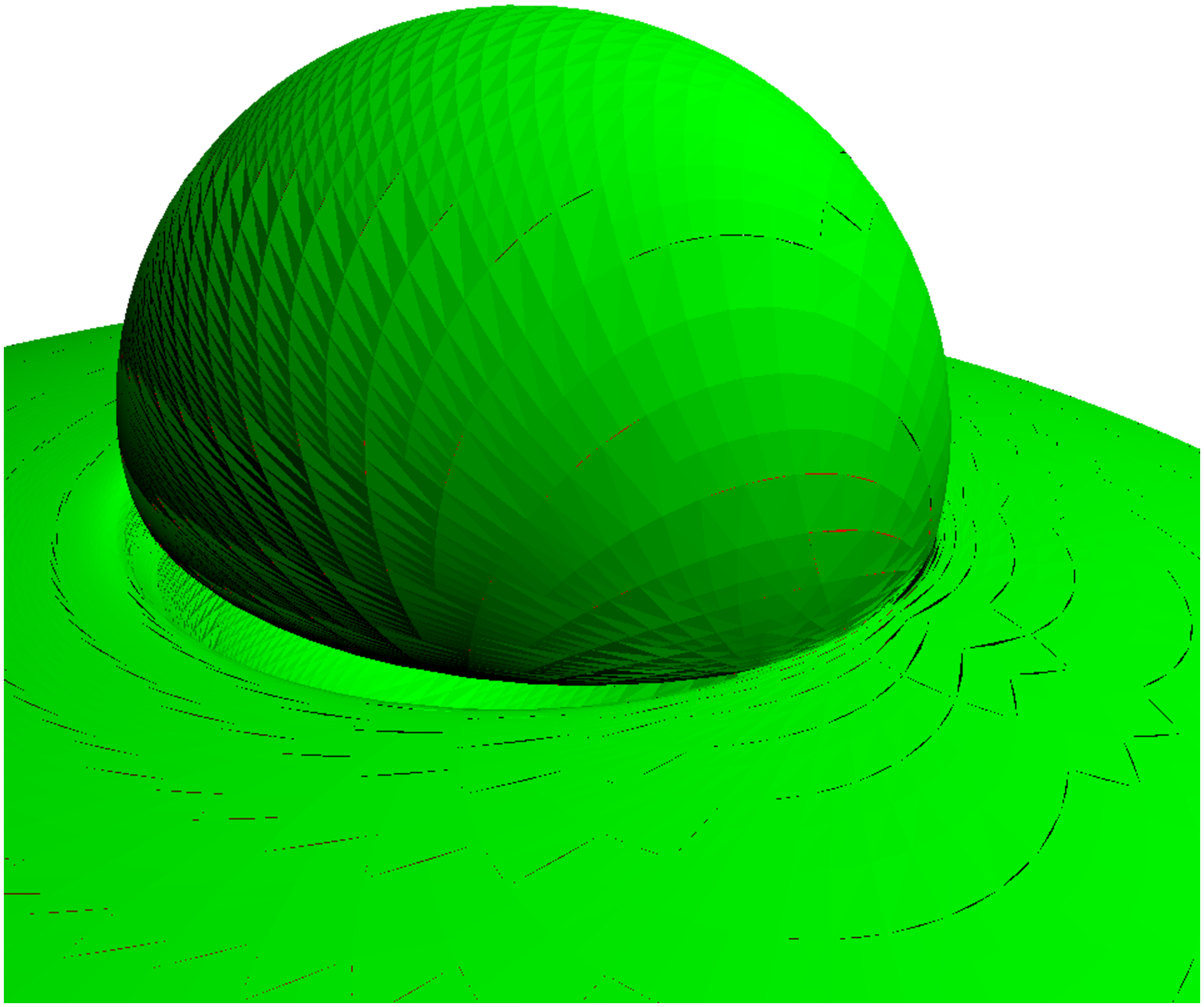}
\end{subfigure}%
\begin{subfigure}{0.5\textwidth}
  \centering
  \includegraphics[width=0.95\textwidth]{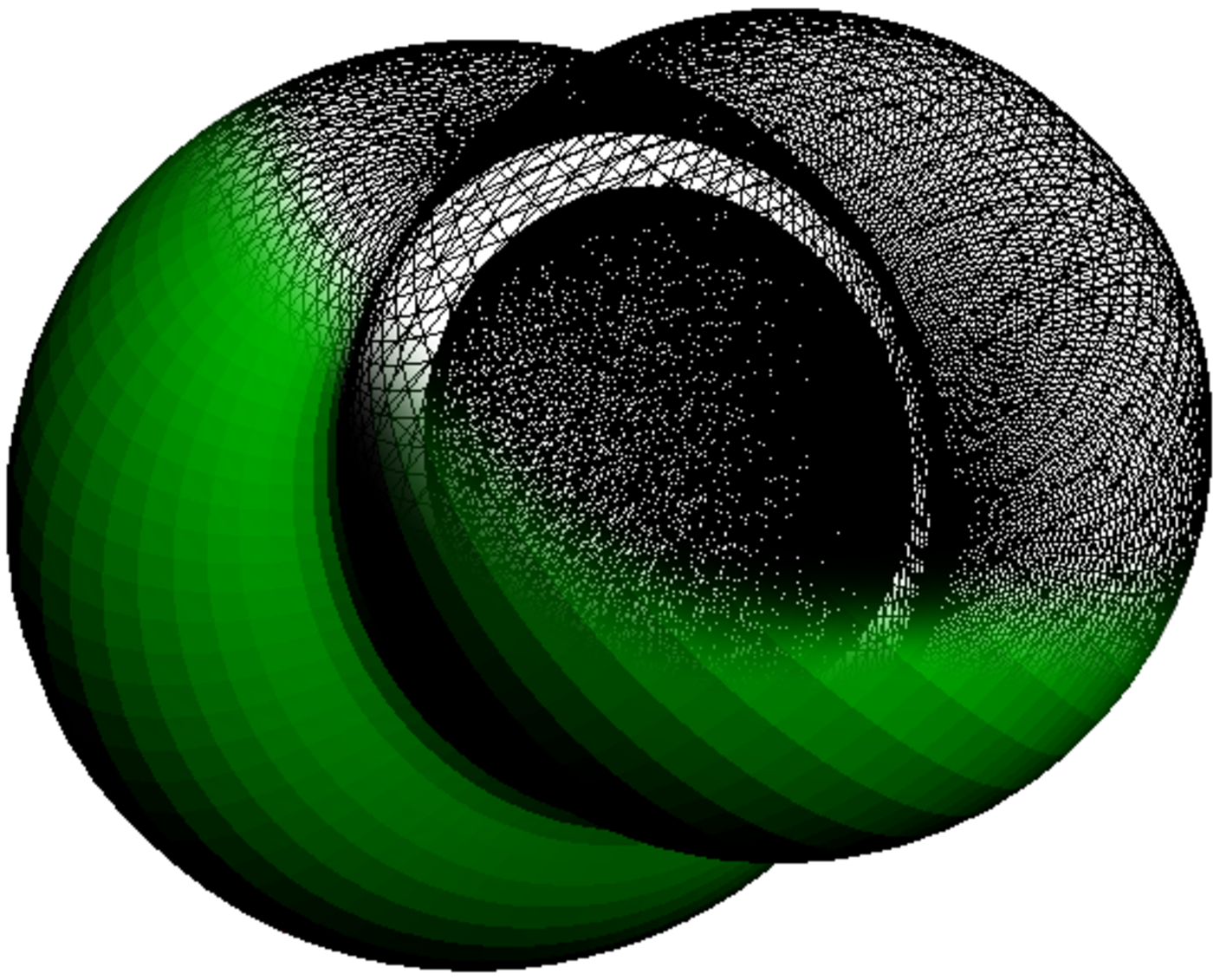}
\end{subfigure}%
\caption{A spin-quadric with a parameterization that is very similar to the $yz$-crossed ellipsoids case (note the topology around the circular edge) and an ellipsoidal parameterization where one of the ellipsoids is inside the other one (the outer ellipsoid is partially transparent)}
\label{fig_not_typical_param}
\end{figure}

\section{Conclusions and future work}

This study showed that the parameterization of configuration space obstacles in the case of three-dimensional rotating objects is feasible as well as practical. In this study, we defined a trigonometric parameterization on a finite domain, which comprises 17 different cases. This expands the number of motion planning problem for which we can parameterize the corresponding configuration space. In fact, this is the first study to parameterize a rotational three-dimensional rotational configuration space with only finite formulae. The previous known method \cite{phd_dobrowolski} includes a number of numerical methods, such as the $LLL$ method for finding rational points on a curve and a numerical method for calculating a matrix determinant.

The next step involves utilizing the parameterization for calculating further geometric properties of the configuration space in three-dimensional rotational motion problems. Thus, the possible uses of the parameterization include the following.
\begin{itemize}
\item Developing a complete motion planner for three-dimensional rotations.
\item Developing an algorithm that uses potential fields for motion planning, where the parameterization can be used to generate gradient fields in the configuration space.
\item Developing a PRM-like algorithm for traversing an intersection graph in the configuration space.
\item Visualizing C-obstacles in $\SO(3)$, especially in conjunction with a stereographic projection $S^3 \longrightarrow \R^3$.
\item Calculating the exact volume of a C-obstacle.
\end{itemize}
The number of different parameterizations with different topologies supports the observation made in \cite{phd_dobrowolski} that the topology of a configuration space in three-dimensional rotational motion planning is geometrically rich and still poorly understood.

\appendix

\section{Proof of the eigenvectors lemma}

\begin{proof}
Let $\lambda_i = c - (\alpha ||P||||Q|| + \beta ||U||V||)$ for $\alpha, \beta \in \{-1, 1\}$ be the eigenvalues of a proper general predicate. Assuming that
\begin{equation}
\lambda_i = c - \gamma_i ,
\end{equation}
where
\begin{equation}
\gamma_i = \alpha ||P||||Q|| + \beta ||U||V|| ,
\end{equation}
then there are also four corresponding $\gamma_i$ values. The elements of the characteristic matrix $F = M_\GG - \lambda_i I$ are as follows.
\begin{equation} \label{eqn_eigen_mtx_general}
F = \left[ \begin{smallmatrix}
        2 (P_3 Q_3 + U_3 V_3) - T + \gamma & 2 (P_1 Q_3 + U_1 V_3) + R_2 & 2 (P_3 Q_2 + U_3 V_2) + R_1 & R_3 \\
        2 (P_1 Q_3 + U_1 V_3) + R_2 & 2 (P_1 Q_1 + U_1 V_1) - T + \gamma & 2 (P_2 Q_1 + U_2 V_1) + R_3 & R_1 \\
        2 (P_3 Q_2 + U_3 V_2) + R_1 & 2 (P_2 Q_1 + U_2 V_1) + R_3 & 2 (P_2 Q_2 + U_2 V_2) - T + \gamma & R_2 \\
        R_3 & R_1 & R_2 & T + \gamma
    \end{smallmatrix} \right]
\end{equation}
Note that $F$ is symmetric. We can rewrite $F$ in brief as follows.
\begin{equation}
F = \left[ \begin{smallmatrix}
        f_{11} & f_{12} & f_{13} & R_3 \\
        f_{12} & f_{22} & f_{23} & R_1 \\
        f_{13} & f_{23} & f_{33} & R_2 \\
        R_3 & R_1 & R_2 & T + \gamma
    \end{smallmatrix} \right]
\end{equation}
An eigenvector can be calculated in different ways according to the order of the row operations. A number of special cases will occur that involve dividing by zero, where all of these are omitted and they are not considered in-place. Instead, the final eigenvector form is checked against the initial eigenvector matrix, and thus it is generalized to all cases.
If $R_3 \ne 0$, the matrix is reduced to
\begin{equation*}
G = \left[ \begin{smallmatrix}
        f_{11} & f_{12} & f_{13} & R_3 \\
        f_{12} R_3 - f_{11} R_1 & f_{22} R_3 - f_{12} R_1 & f_{23} R_3 - f_{13} R_1 & 0 \\
        f_{13} R_3 - f_{11} R_2 & f_{23} R_3 - f_{12} R_2 & f_{33} R_3 - f_{13} R_2 & 0 \\
        R_3 R_3 - f_{11} (T + \gamma) & R_1 R_3 - f_{12} (T + \gamma) & R_2 R_3 - f_{13} (T + \gamma) & 0
    \end{smallmatrix} \right] ,
\end{equation*}
which is denoted in brief as
\begin{equation*}
G = \left[ \begin{smallmatrix}
        f_{11} & f_{12} & f_{13} & R_3 \\
        g_{21} & g_{22} & g_{23} & 0 \\
        g_{31} & g_{32} & g_{33} & 0 \\
        g_{41} & g_{42} & g_{43} & 0
    \end{smallmatrix} \right] ,
\end{equation*}
and if $g_{23} \ne 0$,
\begin{equation*}
H = \left[ \begin{smallmatrix}
        f_{11} & f_{12} & f_{13} & R_3 \\
        g_{21} & g_{22} & g_{23} & 0 \\
        g_{31} g_{23} - g_{21} g_{33} & g_{32} g_{23} - g_{22} g_{33} & 0 & 0 \\
        g_{41} g_{23} - g_{21} g_{43} & g_{42} g_{23} - g_{22} g_{43} & 0 & 0
    \end{smallmatrix} \right] ,
\end{equation*}
which can be shortened to
\begin{equation*}
H = \left[ \begin{smallmatrix}
        f_{11} & f_{12} & f_{13} & R_3 \\
        g_{21} & g_{22} & g_{23} & 0 \\
        h_{31} & h_{32} & 0 & 0 \\
        h_{41} & h_{42} & 0 & 0
    \end{smallmatrix} \right] ,
\end{equation*}
and if $h_{32} \ne 0$,
\begin{equation*}
I = \left[ \begin{smallmatrix}
        f_{11} & f_{12} & f_{13} & R_3 \\
        g_{21} & g_{22} & g_{23} & 0 \\
        h_{31} & h_{32} & 0 & 0 \\
        h_{41} h_{32} - h_{31} h_{42} & 0 & 0 & 0
    \end{smallmatrix} \right] ,
\end{equation*}
which can be shortened to
\begin{equation} \label{eqn_eigen_mtx_ellipsoidal}
I = \left[ \begin{smallmatrix}
        f_{11} & f_{12} & f_{13} & R_3 \\
        g_{21} & g_{22} & g_{23} & 0 \\
        h_{31} & h_{32} & 0 & 0 \\
        i_{41} & 0 & 0 & 0
    \end{smallmatrix} \right] .
\end{equation}
The element $i_{41}$ is a very long expression. By using automatic and manual simplifications, we have the following expression
\begin{align*}
i_{41} = & R_3^2 \left( \|P\|^2 Q_2^2 - \|Q\|^2 P_2^2 + \|U\|^2 V_2^2 - \|V\|^2 U_2^2 \right) \\
         & \left(\gamma^4 - 2 \gamma^2 (\|P\|^2 \|Q\|^2 + \|U\|^2 \|V\|^2) + (\|P\|^2 \|Q\|^2 - \|U\|^2 \|V\|^2)^2 \right),
\end{align*}
which is equal to zero for every eigenvalue $\gamma$
\begin{equation*}
i_{41} = 0.
\end{equation*}
Now, considering only the case of an ellipsoidal general predicate, from Lemma \ref{lem_eigenvalues_of_ellipsoidal_general_predicate}, it is known that each eigenspace is one-dimensional. The fourth row is zero so the first three rows must be linearly independent. From the reduced matrix $I$, an eigenvector $W^{(1)}$ can now be calculated easily, as follows.
\begin{equation*}
W^{(1)} = \left( \begin{smallmatrix}
        R_3 g_{23} h_{32} \\
        - R_3 g_{23} h_{31} \\
        R_3 (g_{22} h_{31} - g_{21} h_{32}) \\
        - f_{11} g_{23} h_{32} + f_{12} g_{23} h_{31} - f_{13} g_{22} h_{31} + f_{13} g_{21} h_{32}
    \end{smallmatrix} \right)
\end{equation*}
The value $R_3$ can be factored out from the fourth coordinate:
\begin{equation*}
W^{(1)} = \left( \begin{smallmatrix}
        R_3 g_{23} h_{32} \\
        - R_3 g_{23} h_{31} \\
        R_3 (g_{22} h_{31} - g_{21} h_{32}) \\
        R_3 (h_{31} (f_{12} f_{23} - f_{13} f_{22}) + h_{32} (f_{13} f_{12} - f_{11} f_{23}))
    \end{smallmatrix} \right) .
\end{equation*}
In this case, $R_3 \ne 0$ and eigenvector $W^{(1)}$ can be scaled to $W^{(2)}$
\begin{equation*}
W^{(2)} = \left( \begin{smallmatrix}
        g_{23} h_{32} \\
        - g_{23} h_{31} \\
        g_{22} h_{31} - g_{21} h_{32} \\
        h_{31} (f_{12} f_{23} - f_{13} f_{22}) + h_{32} (f_{13} f_{12} - f_{11} f_{23})
    \end{smallmatrix} \right) .
\end{equation*}
Using a symbolic calculator, we extract another common factor $\xi$
\begin{equation*}
\xi = R_3 (\|P\|^2 Q_2^2 - \|Q\|^2 P_2^2 + \|U\|^2 V_2^2 - \|V\|^2 U_2^2),
\end{equation*}
which results in $W^{(2)} = [D^{(2)}_3, D^{(2)}_1, D^{(2)}_2, d^{(2)}]^T$, where the vector $D^{(2)}$ and scalar $d^{(2)}$ are equal to
\begin{align*}
D^{(2)} = & \xi \left( - R \gamma^2 + 2 H \gamma + \tilde{T} (\|P\|^2 \|Q\|^2 - \|U\|^2 \|V\|^2) \right) \\
d^{(2)} = & \xi \left( \gamma^3 - T \gamma^2 - (\|\tilde{R}\|^2 + \tilde{T}^2) \gamma + \tilde{T} (\|P\|^2 \|Q\|^2 - \|U\|^2 \|V\|^2) \right),
\end{align*}
where $\tilde{R}$, $\tilde{T}$, and $H$ are defined as follows.
\begin{align} \label{eqn_wide_r_and_wide_t_and_h}
\tilde{R} = & P \times Q - U \times V \\
\tilde{T} = & P \cdot Q - U \cdot V \notag \\
H = & (P \times U) \times (Q \times V) \notag
\end{align}
Again, when $\xi \ne 0$ eigenvector $W^{(2)}$ is scaled to a consecutive representation $W^{(3)} = [D^{(3)}_3, D^{(3)}_1, D^{(3)}_2, d^{(3)}]^T$, where the vector $D^{(3)}$ and scalar $d^{(3)}$ are equal to the following.
\begin{align*}
D^{(3)} = & - R \gamma^2 + 2 H \gamma + \tilde{T} (\|P\|^2 \|Q\|^2 - \|U\|^2 \|V\|^2) \\
d^{(3)} = & \gamma^3 - T \gamma^2 - (\|\tilde{R}\|^2 + \tilde{T}^2) \gamma + \tilde{T} (\|P\|^2 \|Q\|^2 - \|U\|^2 \|V\|^2)
\end{align*}
Further simplifications are possible after $\gamma$ value is substituted explicitly, thereby introducing the variables $\alpha$ and $\beta$ into the formula
\begin{align*}
D^{(3)} = & - 2 \alpha \beta \gamma (\alpha \, \|U\| \|V\| P \times Q + \beta \, \|P\| \|Q\| U \times V \\
          & - \alpha \beta \, (P \times U) \times (Q \times V)) \\
d^{(3)} = & - 2 \alpha \beta \gamma (\alpha \, \|U\| \|V\| P \cdot Q + \beta \, \|P\| \|Q\| U \cdot V \\
          &  - \alpha \beta \, (P \times U) \cdot (Q \times V) - \|P\| \|Q\| \|U\| \|V\|)
\end{align*}
and the eigenvector $W^{(3)}$ is scaled down by the factor $2 \alpha \beta \gamma$ to $W^{(4)} = [D^{(4)}_3, D^{(4)}_1, D^{(4)}_2, d^{(4)}]^T$, where
\begin{align*}
D^{(4)} = &\alpha \beta \, (P \times U) \times (Q \times V) - \alpha \, \|U\| \|V\| P \times Q - \beta \, \|P\| \|Q\| U \times V \\
d^{(4)} = &\alpha \beta \, (P \times U) \cdot (Q \times V) - \alpha \, \|U\| \|V\| P \cdot Q - \beta \, \|P\| \|Q\| U \cdot V + \|P\| \|Q\| \|U\| \|V\| .
\end{align*}
Using normalized vectors, $\hat{P}$, $\hat{Q}$, $\hat{U}$, and $\hat{V}$, the formula can be shortened substantially. The components $D^{(4)}$ and $d^{(4)}$ become
\begin{align*}
D^{(4)} = & \|P\| \|Q\| \|U\| \|V\| (\alpha \beta \, (\hat{P} \times \hat{U}) \times (\hat{Q} \times \hat{V}) - \alpha \, \hat{P} \times \hat{Q} - \beta \, \hat{U} \times \hat{V}) \\
d^{(4)} = & \|P\| \|Q\| \|U\| \|V\| (\alpha \beta \, (\hat{P} \times \hat{U}) \cdot (\hat{Q} \times \hat{V}) - \alpha \, \hat{P} \cdot \hat{Q} - \beta \, \hat{U} \cdot \hat{V} + 1)
\end{align*}
Finally, the eigenvector $W^{(4)}$ is scaled down by the factor $\|P\| \|Q\| \|U\| \|V\|$ to its final form $W = [D_3, D_1, D_2, d]^T$, where
\begin{align*}
D = & \alpha \beta \, (\hat{P} \times \hat{U}) \times (\hat{Q} \times \hat{V}) - \alpha \, \hat{P} \times \hat{Q} - \beta \, \hat{U} \times \hat{V} \\
d = & \alpha \beta \, (\hat{P} \times \hat{U}) \cdot (\hat{Q} \times \hat{V}) - \alpha \, \hat{P} \cdot \hat{Q} - \beta \, \hat{U} \cdot \hat{V} + 1.
\end{align*}
By reinterpreting the formula in terms of Clifford algebra $\Cl(3)$, we find that the eigenvector $W$ is related to the following: \emph{pinor} $\Pin(3) \ni \mathbf{w} = \mathbf{w}_{12} \e_{12} + \mathbf{w}_{23} \e_{23} + \mathbf{w}_{31} \e_{31} + \mathbf{w}_0$
\begin{equation} \label{eqn_ell_eigenvec_in_cl_altern}
\mathbf{w} = 1 + \alpha \beta \, (\hat{P} \times \hat{U}) (\hat{Q} \times \hat{V}) - \alpha \, \hat{P} \hat{Q} - \beta \, \hat{U} \hat{V}
\end{equation}
for an eigenvector $W = [\mathbf{w}_{12}, \mathbf{w}_{23}, \mathbf{w}_{31}, \mathbf{w}_0]^T$. We observe that
\begin{align*}
& (\hat{P} \times \hat{U}) (\hat{Q} \times \hat{V}) = (-\e_{123} \hat{P} \wedge \hat{U}) (-\e_{123} \hat{Q} \wedge \hat{V}) \\
& = - (\hat{P} \wedge \hat{U}) (\hat{Q} \wedge \hat{V}) = - (\hat{P} \cdot \hat{U} + \hat{P} \wedge \hat{U}) (\hat{Q} \cdot \hat{V} + \hat{Q} \wedge \hat{V}) \\
& = - \hat{P} \hat{U} \hat{Q} \hat{V}
\end{align*}
so the pinor $\mathbf{w}$ is equal to
\begin{equation} \label{eqn_ell_eigenvec_in_cl_final}
\mathbf{w} = 1 - \alpha \beta \, \hat{P} \hat{U} \hat{Q} \hat{V} - \alpha \, \hat{P} \hat{Q} - \beta \, \hat{U} \hat{V} .
\end{equation}
In the evaluation, some special cases are omitted so the resulting eigenvector needs to be checked against the initial eigenvalue problem. First, we note that $W$ is valid only for an ellipsoidal general predicate because of the normalized vectors. The toroidal general predicate is considered separately later. Using a symbolic calculator, we can verify that
\begin{equation}
(M_\GG - \lambda I) W = 0,
\end{equation}
so the eigenvector is a correct solution in all ellipsoidal cases. The eigenvector must also be checked to determine whether it has a non-zero magnitude.

A lengthy calculation of the squared length of eigenvector $W$ leads to the following, rather surprising, result.
\begin{equation}
\|W\|^2 = 4 \, W_4
\end{equation}
Clearly, from this result, it follows that the $W_4$ coordinate is always non-negative. It also implies that the eigenvector $W$ is zero if and only if the fourth coordinate $W_4$ is zero. As a result, the eigenvector $W$ is valid if and only if $W_4$ is zero.

\paragraph{Toroidal general predicate case}

Given that the eigenvector $W^{(2)}$ zeroes in the case of a toroidal general predicate, then we must take another approach to calculate the eigenplanes in the latter case.
Consider the toroidal case where $\|U\| \|V\| = 0$. The other case $\|P\| \|Q\| = 0$ is analogous if all the variables $P$ and $Q$ are swapped with $U$ and $V$, respectively. Recall the matrix \eqref{eqn_eigen_mtx_ellipsoidal} and plug in the special case $\|U\| \|V\| = 0$.
\begin{equation} \label{eqn_eigen_mtx_toroidal}
I' = \left[ \begin{smallmatrix}
        f_{11} & f_{12} & f_{13} & R_3 \\
        g_{21} & g_{22} & g_{23} & 0 \\
        h_{31} & h_{32} & 0 & 0 \\
        0 & 0 & 0 & 0
    \end{smallmatrix} \right]
\end{equation}
Using a symbolic calculator, we can verify that
\begin{align*}
h_{31} = & (P_2 Q_1 - P_1 Q_2) (-P_3 Q_2 + P_2 Q_3) (-\gamma^2 + \|P\| \|Q\|^2) = 0 \\
h_{32} = & (P_2 Q_1 - P_1 Q_2)^2 (-\gamma^2 + \|P\|^2 \|Q\|^2) = 0
\end{align*}
for all values if $\gamma$. Hence, the eigenmatrix becomes
\begin{equation} \label{eqn_eigen_mtx_toroidal_reduced}
I' = \left[ \begin{smallmatrix}
        f_{11} & f_{12} & f_{13} & R_3 \\
        g_{21} & g_{22} & g_{23} & 0 \\
        0 & 0 & 0 & 0 \\
        0 & 0 & 0 & 0
    \end{smallmatrix} \right] ,
\end{equation}
where the first and second rows must be independent because the nullspace is two-dimensional.
An eigenplane $Z(u, v)$ can now be calculated as
\begin{equation*}
Z = \left( \begin{smallmatrix}
        -R_3 g_{23} u \\
        -R_3 g_{23} v \\
        R_3 (g_{21} u + g_{22} v) \\
        (f_{11} g_{23} - f_{13} g_{21}) u + (f_{12} g_{23} - f_{13} g_{22}) v
    \end{smallmatrix} \right)
\end{equation*}
and it is possible to extract a common factor $R_3$ from the fourth coordinate, as follows.
\begin{equation*}
Z = R_3 \left( \begin{smallmatrix}
        -g_{23} u \\
        -g_{23} v \\
        g_{21} u + g_{22} v \\
        (f_{11} f_{23} - f_{12} f_{13}) u + (f_{12} f_{23} - f_{13} f_{22}) v
    \end{smallmatrix} \right)
\end{equation*}
By using a symbolic calculator, we find the following.
\begin{align} \label{eqn_eigen_mtx_toroidal_gs}
g_{21} = - P_3 P_2 \|Q\|^2 + Q_3 Q_2 \|P\|^2 - \gamma (P_2 Q_3 - P_3 Q_2) \\
g_{22} = - P_1 P_2 \|Q\|^2 + Q_1 Q_2 \|P\|^2 - \gamma (P_2 Q_1 - P_1 Q_2) \notag \\
g_{23} = - P_2 P_2 \|Q\|^2 + Q_2 Q_2 \|P\|^2 - \gamma (P_2 Q_2 - P_2 Q_2) \notag
\end{align}
In this case, we note that
\begin{equation*}
\gamma = \alpha \|P\| \|Q\| ,
\end{equation*}
and thus \eqref{eqn_eigen_mtx_toroidal_gs} simplifies to
\begin{align} \label{eqn_eigen_mtx_toroidal_gs_simp}
g_{21} = \|P\|^2 \|Q\|^2 \left(- \hat{P}_3 \hat{P}_2 + \hat{Q}_3 \hat{Q}_2 - \alpha (\hat{P}_2 \hat{Q}_3 - \hat{P}_3 \hat{Q}_2) \right) \\
g_{22} = \|P\|^2 \|Q\|^2 \left(- \hat{P}_1 \hat{P}_2 + \hat{Q}_1 \hat{Q}_2 - \alpha (\hat{P}_2 \hat{Q}_1 - \hat{P}_1 \hat{Q}_2) \right) \notag \\
g_{23} = \|P\|^2 \|Q\|^2 \left(- \hat{P}_2 \hat{P}_2 + \hat{Q}_2 \hat{Q}_2 - \alpha (\hat{P}_2 \hat{Q}_2 - \hat{P}_2 \hat{Q}_2) \right) \notag.
\end{align}
If $R_3 \|P\|^2 \|Q\|^2 \ne 0$, the eigenplane can be scaled down to $Z^{(2)}$
\begin{align*}
Z^{(2)} = & \left( \begin{smallmatrix}
        \hat{P}_2 \hat{P}_2 - \hat{Q}_2 \hat{Q}_2 - \alpha (\hat{P}_2 \hat{Q}_2 - \hat{P}_2 \hat{Q}_2) \\
        0 \\
        - \hat{P}_3 \hat{P}_2 + \hat{Q}_3 \hat{Q}_2 + \alpha (\hat{P}_3 \hat{Q}_2 - \hat{P}_2 \hat{Q}_3) \\
        - \hat{P}_1 \hat{P}_2 - \hat{Q}_1 \hat{Q}_2 + \alpha (\hat{P}_1 \hat{Q}_2 + \hat{P}_2 \hat{Q}_1)
    \end{smallmatrix} \right) u
    + \left( \begin{smallmatrix}
        0 \\
        \hat{P}_2 \hat{P}_2 - \hat{Q}_2 \hat{Q}_2 - \alpha (\hat{P}_2 \hat{Q}_2 - \hat{P}_2 \hat{Q}_2) \\
        - \hat{P}_1 \hat{P}_2 + \hat{Q}_1 \hat{Q}_2 + \alpha (\hat{P}_1 \hat{Q}_2 - \hat{P}_2 \hat{Q}_1) \\
        \hat{P}_3 \hat{P}_2 + \hat{Q}_3 \hat{Q}_2 - \alpha (\hat{P}_3 \hat{Q}_2 + \hat{P}_2 \hat{Q}_3)
    \end{smallmatrix} \right) v
\end{align*}
We estimate that an analogous formula should exist that relates the two remaining coordinates. Using a symbolic calculator, the following two formulae can be verified. Formulae
\begin{align*}
Z^{(1)} = & \left( \begin{smallmatrix}
        \hat{P}_2 \hat{P}_1 - \hat{Q}_2 \hat{Q}_1 - \alpha (\hat{P}_2 \hat{Q}_1 - \hat{P}_1 \hat{Q}_2) \\
        0 \\
        - \hat{P}_3 \hat{P}_1 + \hat{Q}_3 \hat{Q}_1 + \alpha (\hat{P}_3 \hat{Q}_1 - \hat{P}_1 \hat{Q}_3) \\
        - \hat{P}_1 \hat{P}_1 - \hat{Q}_1 \hat{Q}_1 + \alpha (\hat{P}_1 \hat{Q}_1 + \hat{P}_1 \hat{Q}_1)
    \end{smallmatrix} \right) u
    + \left( \begin{smallmatrix}
        0 \\
        \hat{P}_2 \hat{P}_1 - \hat{Q}_2 \hat{Q}_1 - \alpha (\hat{P}_2 \hat{Q}_1 - \hat{P}_1 \hat{Q}_2) \\
        - \hat{P}_1 \hat{P}_1 + \hat{Q}_1 \hat{Q}_1 + \alpha (\hat{P}_1 \hat{Q}_1 - \hat{P}_1 \hat{Q}_1) \\
        \hat{P}_3 \hat{P}_1 + \hat{Q}_3 \hat{Q}_1 - \alpha (\hat{P}_3 \hat{Q}_1 + \hat{P}_1 \hat{Q}_3)
    \end{smallmatrix} \right) v
\end{align*}
and
\begin{align*}
Z^{(3)} = & \left( \begin{smallmatrix}
        \hat{P}_2 \hat{P}_3 - \hat{Q}_2 \hat{Q}_3 - \alpha (\hat{P}_2 \hat{Q}_3 - \hat{P}_3 \hat{Q}_2) \\
        0 \\
        - \hat{P}_3 \hat{P}_3 + \hat{Q}_3 \hat{Q}_3 + \alpha (\hat{P}_3 \hat{Q}_3 - \hat{P}_3 \hat{Q}_3) \\
        - \hat{P}_1 \hat{P}_3 - \hat{Q}_1 \hat{Q}_3 + \alpha (\hat{P}_1 \hat{Q}_3 + \hat{P}_3 \hat{Q}_1)
    \end{smallmatrix} \right) u
    + \left( \begin{smallmatrix}
        0 \\
        \hat{P}_2 \hat{P}_3 - \hat{Q}_2 \hat{Q}_3 - \alpha (\hat{P}_2 \hat{Q}_3 - \hat{P}_3 \hat{Q}_2) \\
        - \hat{P}_1 \hat{P}_3 + \hat{Q}_1 \hat{Q}_3 + \alpha (\hat{P}_1 \hat{Q}_3 - \hat{P}_3 \hat{Q}_1) \\
        \hat{P}_3 \hat{P}_3 + \hat{Q}_3 \hat{Q}_3 - \alpha (\hat{P}_3 \hat{Q}_3 + \hat{P}_3 \hat{Q}_3)
    \end{smallmatrix} \right) v
\end{align*}
also define a valid eigenplane. For each of the three formulae, parameters exist where each of the eigenplanes zeroes. If we suppose that none of the three formulae spans a valid eigenplane, then the first two coordinates never form a linearly independent base and the following equations hold
\begin{align*}
\begin{cases}
& \hat{P}_2 \hat{P}_1 - \hat{Q}_2 \hat{Q}_1 - \alpha (\hat{P}_2 \hat{Q}_1 - \hat{P}_1 \hat{Q}_2) = 0 \\
& \hat{P}_2 \hat{P}_2 - \hat{Q}_2 \hat{Q}_2 - \alpha (\hat{P}_2 \hat{Q}_2 - \hat{P}_2 \hat{Q}_2) = 0 \\
& \hat{P}_2 \hat{P}_3 - \hat{Q}_2 \hat{Q}_3 - \alpha (\hat{P}_2 \hat{Q}_3 - \hat{P}_3 \hat{Q}_2) = 0
\end{cases} ,
\end{align*}
which are equivalent to
\begin{equation} \label{eqn_eigen_mtx_toroidal_spec}
\hat{P} \, (\hat{P}_2 + \alpha \hat{Q}_2) = \hat{Q} \, (\hat{Q}_2 + \alpha \hat{P}_2) \iff (\hat{P} - \alpha \hat{Q}) (\hat{P}_2 + \alpha \hat{Q}_2) = 0.
\end{equation}
Now, we consider the special cases where $\hat{P} = \alpha \hat{Q}$ and $\hat{P}_2 = -\alpha \hat{Q}_2$. In the first case, eigenmatrix \eqref{eqn_eigen_mtx_general} becomes
\begin{equation*}
2 \alpha \|P\|^2 \left[
    \begin{smallmatrix}
        \hat{Q}_3 \hat{Q}_3 & \hat{Q}_3 \hat{Q}_1 & \hat{Q}_3 \hat{Q}_2 & 0 \\
        \hat{Q}_1 \hat{Q}_3 & \hat{Q}_1 \hat{Q}_1 & \hat{Q}_1 \hat{Q}_2 & 0 \\
        \hat{Q}_2 \hat{Q}_3 & \hat{Q}_2 \hat{Q}_1 & \hat{Q}_2 \hat{Q}_2 & 0 \\
        0 & 0 & 0 & 1
    \end{smallmatrix}
\right]
\end{equation*}
 and there are two eigenplanes, where the first is spanned by the vectors
\begin{align*}
\left( \begin{smallmatrix}
    0 \\
    0 \\
    0 \\
    1
\end{smallmatrix} \right) u +
\left( \begin{smallmatrix}
    \hat{Q}_3 \\
    \hat{Q}_1 \\
    \hat{Q}_2 \\
    0
\end{smallmatrix} \right) v
\end{align*}
and the second is spanned by the vectors
\begin{align*}
\left( \begin{smallmatrix}
    -\hat{Q}_2 \\
    0 \\
    \hat{Q}_3 \\
    0
\end{smallmatrix} \right) u +
\left( \begin{smallmatrix}
    -\hat{Q}_1 \\
    \hat{Q}_3 \\
    0 \\
    0
\end{smallmatrix} \right) v.
\end{align*}
We have considered all the special cases so the proof is complete.
\end{proof}


\bibliographystyle{elsarticle-num}
\bibliography{references}

\end{document}